\documentclass[11pt]{article}

\usepackage{amsbsy, graphicx, amssymb, amsmath, amsthm, amsfonts, mathrsfs, color}
\usepackage[round]{natbib}
\usepackage[colorlinks, citecolor=blue, urlcolor=blue, linkcolor=blue, breaklinks]{hyperref}

\newtheorem{theorem}{Theorem}

\newtheorem{lemma}{Lemma}

\newtheorem{remark}{\bf Remark}

 \allowdisplaybreaks

\newcommand{\xar}{\xrightarrow}
\newcommand{\eps}{\varepsilon}
\newcommand{\vpt}{\varphi, \vartheta}
\DeclareMathOperator*{\argmin}{argmin}
\newcommand{\md}{\mathrm{d}}

\newcommand{\sumtn}{\sum_{t=1}^n}

\newcommand{\Ft}{\mathcal{F}_{t-1}}

\newcommand{\X}{\mathbf{X}}

\newcommand{\Z}{\mathbf{Z}}

\newcommand{\topr}{\xar[]{p}}
\newcommand{\op}{o_p}
\newcommand{\Op}{O_p}

\begin{document}

\begin{center}
 {\bf\LARGE  Quantile Regression for Location-Scale Time Series
 Models with Conditional Heteroscedasticity}
 \vspace{.5cm}
\end{center}
\begin{center}
Jungsik Noh\footnote{Quantitative Biomedical Research Center, Department of Clinical Sciences, University of Texas Southwestern Medical Center, Dallas, TX 75390, USA. Email: nohjssunny@gmail.com}
and
Sangyeol Lee\footnote{Department of Statistics, Seoul National
University, Seoul 151-747, Korea. Email: sylee@stats.snu.ac.kr}\\
$^1$University of Texas Southwestern Medical Center \\
$^2$Seoul National University
\end{center}

\begin{center}
    Revised February 28, 2015
\end{center}

\begin{abstract}
    { This paper considers quantile regression for a wide class of time series models
    including ARMA models with asymmetric GARCH (AGARCH) errors.
    The classical mean-variance models are reinterpreted as conditional location-scale
    models so that the quantile regression method can be naturally
    geared into the considered models.
    The consistency and asymptotic normality of the quantile regression
    estimator is established in location-scale time series models under mild conditions.
    In the application of this result to ARMA-AGARCH models, more primitive conditions are deduced to obtain the asymptotic
    properties. For illustration, a simulation study and a real data
    analysis are provided.
    }
\end{abstract}

\noindent{\bf MSC2010 subject classifications}: Primary 62M10; secondary 62F12.
\\
\noindent{\bf Abbreviated title}: Quantile regression for
location-scale time series models

\section{Introduction}

Quantile regression, introduced by \citet{KoenkerBassett1978},
generalizes the notion of sample quantiles to linear and nonlinear
regression models including the least absolute deviation estimation
as its special case. The method provides an estimation of
conditional quantile functions at any probability levels and it is
well known that the family of estimated conditional quantiles sheds
a new light on the impact of covariates on the conditional location,
scale and shape of the response distribution: see
\citet{Koenker2000}.
Quantile regression has been widely used to analyze time series data
as an alternative
to the least squares method (see  \citealt{Fitzenbergeretal2002,
Koenker2005}) since it is not only robust to heavy tails but also
allows a flexible analysis of the covariate effects. Especially, in
risk management, it is also a functional tool to calculate the
value-at-risk (VaR). Quantile regression has been studied in linear
and nonlinear autoregressive models by
\citet{BloomfieldSteiger1983}, \citet{Weiss1991},
\citet{KoulSaleh1995}, and \citet{DavisDunsmuir1997}: see also
\citet{KoenkerZhao1996} and \citet{XiaoKoenker2009}, who handled
`linear' autoregressive conditional heteroscedasticity (ARCH) and
generalized ARCH (GARCH) models, and \citet{LeeNoh2013} who
considered ordinary GARCH models. \citet{EngleManganelli2004}
considered the quantile regression method for a broad class of time
series models and designated the conditional autoregressive VaR
(CAViaR) model. Although the results of \citet{EngleManganelli2004}
are applicable to a wide class of time series models, the CAViaR
specification therein mainly focuses on the case of pure volatility
models, as pointed out by \citet{Kuesteretal2006} and
\citet{Schaumburg2012}. Unlike the previous studies dealing with the
models having either conditional location or scale components, in
this study, we take an approach to simultaneously estimate the
conditional mean and variance through the quantile regression
method.
\citet{KoenkerBassett1982} and \citet{KoenkerXiao2002} explored the
quantile regression for location-scale models without autoregressive
structure and proposed a robust test for heteroscedasticity.


This paper focuses on the quantile regression for a wide class of
conditional location-scale time series models including the ARMA
models with asymmetric GARCH (AGARCH) errors in which the dynamic
relation between current and past observations is characterized in
terms of a conditional mean and variance structure. Typically, the
conditional mean is assumed to follow an either AR or ARMA type
model and the conditional volatility is assumed to follow a GARCH
type model (\citealt{Bollerslev2008}). Here, we demonstrate that the
quantile regression can be extended to conditional location-scale
models rather than mean-variance models through a slight
modification, and as such, the estimation of the conditional
location and scale can be properly carried out. More precisely, to
activate the proposed method, we remove the constraints imposed on
the mean and variance of the model innovations and reformulate the
mean-variance model to become the conditional location-scale model
described in {Section}~\ref{subsec22}.
It is noteworthy that the reformulated models to incur the quantile
regression estimation are exactly the same as those in (1.3) of
\citet{NeweySteigerwald1997} who pointed out that non-Gaussian
quasi-maximum likelihood (QML) estimators may be inconsistent in the
usual conditional mean-variance models and instead proposed
location-scale models to remedy an asymptotic bias effect. From this
angle, it may be mentioned that our quantile regression method is
comparable with other estimation methods like the Gaussian and
non-Gaussian QML estimation methods.

In this study, we intend to verify the strong consistency and
asymptotic normality of quantile regression estimators in  general
conditional location-scale time series models. Particularly, in the
derivation of the $\sqrt{n}$-consistency, one has to overcome the
difficulty caused by the lack of smoothness of the quantile
regression loss function. To resolve this problem, we adopt the idea
of \citet{Huber1967} and \citet{Pollard1985} and extend Lemma~3 of
\citet{Huber1967} to stationary and ergodic time series cases; see
{Section}~\ref{subsec23} and Lemma~\ref{BR lemma} in the Appendix
for details. To apply the obtained results in general models to the
ARMA-AGARCH model, we deduce certain primitive conditions leading to
the desired asymptotic properties. Here, the task of checking the
identifiability condition appears to be quite demanding and
accordingly a newly designed technique is proposed: see
Remark~\ref{identifiability proof remark} below.

In comparison to \citet{EngleManganelli2004}, our approach has merit
in its own right. First, a weaker moment condition is used to obtain
the asymptotic normality:
for instance, in the ARMA-AGARCH model, only a finite second moment
condition is required while a third moment condition is demanded in
their paper.
Second, more basic conditions such as strict stationarity and
ergodicity of models are assumed in our case rather than the law of
large numbers and central limit theorems assumed in their paper:
however, more general data generating processes are considered
therein.
Third, our parametrization of conditional quantile functions
exhibits a more explicit relationship with the parametrization of
original models.
Finally, a general identifiability condition is provided for the
ARMA-AGARCH model and is rigorously verified.

The rest of this article is organized as follows. In
Section~\ref{sec2}, we introduce the general conditional
location-scale time series models and establish the asymptotic
properties of the quantile regression estimator. In
Section~\ref{sec3}, we verify the conditions for the strong
consistency and asymptotic normality in the ARMA-AGARCH model. In
Section~\ref{sec4}, we report a finite sample performance of the
estimator in comparison with the Gaussian-QMLE. In
Section~\ref{sec5}, we demonstrate the validity of our method by
analyzing the daily returns of Hong Kong Hang Seng Index. All the
proofs are provided in the Appendix and the supplementary material.

\section{Quantile regression estimation of conditional heteroscedasticity models}\label{sec2}

\subsection{An example: reparameterized AR($1$)-ARCH($1$) models}\label{subsec21}

Before we proceed to general conditional location-scale models (see
(\ref{loc-scale model}) below), we first illustrate conditional
quantile estimation for the AR($1$)-ARCH($1$) model:
\begin{align}\label{AR-ARCH eqn}
\begin{aligned}
    Y_t &= \phi_0^\circ + \phi_1^\circ Y_{t-1} + \eps_t, \\
    \eps_t &= \sigma_t \eta_t, \quad  \sigma_t^2 = \omega^\circ + \alpha^\circ \eps_{t-1}^2 \quad \text{for
    }t\in\mathbb{Z},
\end{aligned}
\end{align}
where $\{\eta_t\}$ are i.i.d. random variables with $E\eta_t=0$ and
$E\eta_t^2=1$. In what follows, we denote by
$\mathcal{F}_{t}=\sigma\left( Y_s:s\leq t\right)$ the $\sigma$-field
generated by $\{Y_{s}:s\leq t\}$. Provided that $\eta_t$ is
independent of $\Ft$, the $\tau$th conditional quantile of $Y_t$
given $\Ft$ can be expressed as
\begin{align}\label{cond quantile of AR-ARCH}
    Q_\tau(Y_t |\Ft) &= \phi_0^\circ + \phi_1^\circ Y_{t-1} + F^{-1}_\eta (\tau)
    \left( \omega^\circ + \alpha^\circ \left(Y_{t-1} - \phi_0^\circ - \phi_1^\circ Y_{t-2}\right)^2
    \right)^{1/2},
\end{align}
where $F^{-1}_\eta(\tau)=\inf\{x: P(\eta_1\leq x)\geq \tau \}$. 
Since the $\tau$th quantile of $\eta_t$ is unknown, it is apparent
that the parameters in (\ref{cond quantile of AR-ARCH}) are not
identifiable. As in \citet{LeeNoh2013}, this problem can be overcome
by reparameterizing the ARCH component as follows:
\begin{align*}
    \eps_t &=  h_t u_t, \quad  h_t^2 =1+\gamma^\circ \eps_{t-1}^2
\end{align*}
with $h_t^2 = \sigma_t^2 /{\omega^\circ}$,
$u_t=\sqrt{\omega^\circ}\eta_t$ and
$\gamma^\circ=\alpha^\circ/\omega^\circ$. Here, $h_t$ is only
proportional  to the conditional standard deviation, and thus, can
be interpreted to be a conditional scale: this reparameterization
procedure expresses the ARCH model as a conditional scale model with
no scale constraints on the i.i.d. innovations. The conditional
quantile in this case is then expressed as \begin{align}\label{cond
quantile of repra AR-ARCH}
    Q_\tau(Y_t |\Ft) &= \phi_0^\circ + \phi_1^\circ Y_{t-1} + F^{-1}_u(\tau)
    \left( 1 + \gamma^\circ \left(Y_{t-1} - \phi_0^\circ - \phi_1^\circ Y_{t-2}\right)^2
    \right)^{1/2},
\end{align}
wherein the parameters can be shown to be identifiable: see
Lemma~\ref{identification lemma} that deals with more general
ARMA-AGARCH models.

In fact, the condition $E\eta_t=0$ in (\ref{AR-ARCH eqn}) is not
necessarily required to deduce the conditional quantile function,
since the conditional quantile specification in (\ref{cond quantile
of repra AR-ARCH}) is also valid for the AR($1$)-ARCH($1$) model
without assuming this condition. As seen in Section~\ref{subsec23},
conditional quantile estimators and their asymptotic properties are
irrelevant to the location constraint on $u_t$, and thus, the
condition of $E u_t=0$ is not needed for estimating conditional
quantiles. An analogous approach will be taken to handle the
quantile regression for general location-scale models.

\subsection{Conditional location-scale models}\label{subsec22}

Let us consider the general conditional location-scale model of the
form:
\begin{align}\label{loc-scale model}
\begin{aligned}
    Y_t = f_t(\alpha^\circ) + h_t(\alpha^\circ) u_t \quad \text{for }t\in\mathbb{Z},
\end{aligned}
\end{align}
where $f_t(\alpha^\circ)$ and $h_t(\alpha^\circ)$ respectively
denote $f(Y_{t-1}, Y_{t-2}, \ldots;\alpha^\circ)$ and $h(Y_{t-1},
Y_{t-2}, \ldots;\alpha^\circ)$ for some measurable functions $f,
h:\mathbb{R}^\infty\times\Theta_1 \to\mathbb{R}$; $\alpha^\circ$
denotes the true model parameter; $\Theta_1$ is a model parameter
space; $\{u_t\}$ are i.i.d. random variables with an unknown common
distribution function $F_u$.

Many conditionally heteroscedastic time series models can be
described by the autoregressive representation addressed in
(\ref{loc-scale model}). For example, the reparameterized
AR($1$)-ARCH($1$) model in Section~\ref{subsec21} can be expressed
as a form of (\ref{loc-scale model}) with $\alpha=(\phi_0, \phi_1,
\gamma)$, $f_t(\alpha)=\phi_0 + \phi_1 Y_{t-1}$ and $h_t^2(\alpha) =
1 + \gamma \left(Y_{t-1} - \phi_0 - \phi_1 Y_{t-2}\right)^2 $.
Further, it can be readily seen that invertible ARMA models and
stationary GARCH models also admit the form of (\ref{loc-scale
model}): see Theorem~2.1 of \citet{Berkesetal2003} for the latter.
In Section~\ref{sec3}, the ARMA-AGARCH model will be expressed as a
form of (\ref{loc-scale model}).

{In order to facilitate the conditional quantile estimation,
model~(\ref{loc-scale model}) is assumed to be a reparameterized
version of the time series models as discussed in
Section~\ref{subsec21}, and as such, the innovation distribution
$F_u$ is not assumed to have zero mean and unit variance and
$h_t(\alpha^\circ)$ is interpreted to be a relative conditional
scale rather than variance.  However, restricted to ARMA-AGARCH
models in Sections~\ref{sec3}--\ref{sec5}, we focus on the case of
$Eu_t=0$ considering the popularity in practice.}

In what follows, the following conditions are presumed:
\begin{itemize}
  \item[{\bf (M1)}] $\{Y_t:t\in\mathbb{Z}\}$ satisfying (\ref{loc-scale model}) is strictly stationary and ergodic.
  \item[{\bf (M2)}] $u_t$ is independent of $\mathcal{F}_s$ for $s<t$.
\end{itemize}

 {Conditions {\bf (M1)} and {\bf (M2)} hold for a broad class of time series models. For example,
\citet{BougerolPicard1992} verified that the GARCH model is strictly
stationary if and only if its Lyapunov exponent is negative, which
actually entails {\bf (M2)}. \citet{StraumannMikosch2006} provided
sufficient conditions for the stationarity and ergodicity in general
conditional variance models. \citet{MeitzSaikkonen2008} provided
such conditions in nonlinear AR models with GARCH errors. In
Section~\ref{sec3}, we  specify some conditions for the ARMA-AGARCH
model to admit the autoregressive representation in (\ref{loc-scale
model}) and also to satisfy  {\bf (M1)} and {\bf (M2)}. }

Under {\bf (M2)}, the $\tau$th
quantile of $Y_t$ conditional on the past observations is given by $
 Q_\tau(Y_t |\Ft) =f_t(\alpha^\circ) + \xi^\circ(\tau) h_t(\alpha^\circ)
$ for $0<\tau<1$, wherein
 {the innovation quantile}
$\xi^\circ(\tau) = F_u^{-1}(\tau)$ appears as a new parameter. We
denote by $\theta^\circ(\tau)=(\xi^\circ(\tau), {\alpha^\circ}^T)^T$
the true parameter vector.
 {Note that the conditional quantile}
$Q_\tau(Y_t |\Ft)$ can be expressed as a function of the infinite
number of past observations and parameter $\theta^\circ(\tau)$.
Then, taking into consideration the form of $Q_\tau(Y_t |\Ft)$,
given the stationary solution $\{Y_t\}$ to model~(\ref{loc-scale
model}) and a parameter vector $\theta=(\xi, \alpha^T)^T$,  we
introduce conditional quantile functions:
\begin{align}\label{def of q_t(theta)}
    q_t(\theta)= f(Y_{t-1}, Y_{t-2}, \ldots;\alpha) + \xi h(Y_{t-1}, Y_{t-2}, \ldots;\alpha) \quad \text{for }t\in\mathbb{Z},
\end{align}
where $\alpha$ is a parameter within a domain that allows the above
autoregressive representation. In practice, since $\{Y_t : t\leq0 \}
$ is unobservable, we cannot obtain $q_t(\theta), 1\leq t\leq n$.
Thus, we approximate them with observable $\tilde{q}_t(\theta),
1\leq t\leq n$. A typical example is
$\tilde{q}_t(\theta)=f(Y_{t-1},\ldots, Y_{1},0, \ldots;\alpha)+\xi
h(Y_{t-1},\ldots, Y_{1},0, \ldots;\alpha)$, where all $Y_t$ with
$t\leq 0$ are put to be 0: see \citet{Panetal2008}. One can also use
a model specific approximation as in Section~\ref{sec3}.
Then, the $\tau$th quantile regression estimator of
$\theta^\circ(\tau)$ for model~(\ref{loc-scale model}) is defined by
\begin{align}\label{def of QRE}
    \hat{\theta}_n(\tau) = \argmin_{\theta\in\Theta} {1\over n} \sum_{t=1}^n \rho_\tau(Y_t - \tilde{q}_t(\theta)),
\end{align}
where $\Theta\subset\mathbb{R}^d$ is a parameter space,
$\rho_\tau(u)=u(\tau-I(u<0))$,
 {
and $I(\cdot)$ denotes the indicator function.
}

\subsection{Asymptotic properties of quantile regression estimators}\label{subsec23}

In this subsection, we  {show} the strong consistency and asymptotic
normality of the quantile regression estimator defined in (\ref{def
of QRE}). {The result is applicable to various mean-variance time
series models including the ARMA-AGARCH model handled in
Section~\ref{sec3}.}
 {The} asymptotic properties
are proved by utilizing the affinity between $q_t(\cdot)$ and
$\tilde{q}_t(\cdot)$,  similarly to the {case of the QML estimator
in GARCH-type models: see \citet{Berkesetal2003},
\citet{FrancqZakoian2004}, \citet{StraumannMikosch2006},
\citet{LeeLee2012}, and the references therein. However, the
asymptotic normality is derived in a nonstandard situation, as
discussed below, owing to the non-differentiability of the
 {loss} function $\rho_\tau(\cdot)$.}

  In what follows, we define $\|A\|=\sum_{i,j}|a_{ij}|$ for  matrix $A=(a_{ij})$.
To verify the consistency of $\hat\theta_n(\tau)$, we introduce the
following assumptions:
\begin{description}
    \item [(C1)] The $\tau$th quantile of $u_1$ is unique, that is, $F_u(F_u^{-1}(\tau)-\eps)<\tau< F_u(F_u^{-1}(\tau)+\eps)$ for all $\eps>0$.
    \item[(C2)] $\theta^\circ(\tau)$ belongs to $\Theta$ which is a compact subset of $\mathbb{R}^d$.
    \item[(C3)] (i) $q_1(\theta)$ is continuous in $\theta\in\Theta$  {a.s.};
                 {  (ii) $E\left[ \sup_{\theta\in\Theta} \left| q_1(\theta) \right| \right] <\infty $.  }
    \item[(C4)] If $q_t(\theta)=q_t(\theta^\circ(\tau))$ a.s. for some $t\in\mathbb{Z}$ and $\theta\in\Theta$, then $\theta=\theta^\circ(\tau)$.
    \item[(C5)] There exists a positive constant $c_0$  such that $ h_1(\alpha^\circ)\geq c_0$  {a.s}.
    \item[(C6)]   {$\sum_{t=1}^\infty \sup_{\theta\in\Theta} \left| q_t(\theta) -\tilde{q}_t(\theta) \right| <\infty$ a.s.}
\end{description}

\begin{theorem}\label{strong consistency}
    Suppose that  assumptions {\bf (M1)}, {\bf (M2)}, and {\bf (C1)--(C6)} hold for model~(\ref{loc-scale
    model}).
    Then,  $\hat{\theta}_n(\tau) \to \theta^\circ(\tau)$ a.s. as $n\to\infty$.
\end{theorem}

 {It can be seen that $\{q_t(\theta):t\in\mathbb{Z}\}$ in
(\ref{def of q_t(theta)}) is strictly stationary and ergodic for
each $\theta$ (see Proposition~2.5 of
\citealt{StraumannMikosch2006}), while its approximation
$\{\tilde{q}_t(\theta):t\in\mathbb{N}\}$ in (\ref{def of QRE}) is
not so since $\tilde{q}_t(\theta)$ is recursively defined with given
initials. Assumptions {\bf(C2)}, {\bf(C3)} and {\bf(C6)} are needed
to show the uniform convergence of the objective function in
(\ref{def of QRE}),
 based on the ergodic theorem   
of \citet{StraumannMikosch2006}. It can be shown from assumptions
{\bf(C1)}, {\bf(C4)} and {\bf(C5)} that the a.s. limit of the
objective function is uniquely minimized at
$\theta=\theta^\circ(\tau)$. Conventionally, all the assumptions
excepting the identifiability assumption {\bf (C4)} are easy to
check from the existing literatures. The ARMA-AGARCH model can be
shown to satisfy these conditions in Section~\ref{sec3}, based on
\citet{BrockwellDavis1991} and \citet{Panetal2008}. }

 {Below, we discuss two issues as to {\bf (C4)}.
First, when $\xi^\circ(\tau)=F_u^{-1}(\tau)=0$,
 {\bf (C4)} is not satisfied for heteroscedastic models.
 It is because the parameters involved in
$h_t(\alpha^\circ)$ is not identifiable as seen in the GARCH model:
see Remark~3 of \citet{LeeNoh2013}. If $\xi^\circ(\tau)=0$, only the
parameters in the conditional location is estimable via using
\citet{Weiss1991} who proposed the conditional median estimation for
some models similar to (\ref{loc-scale model}) when
$\xi^\circ(0.5)=0$. This indicates that the $\tau$th conditional
quantile estimation for heteroscedastic models requires a different
conditional quantile specification at some $\tau$, usually the one
corresponding to a center of locations.}

{Secondly, the verification of {\bf (C4)} is non-standard even in
the AR($1$)-ARCH($1$) model, see (\ref{cond quantile of repra
AR-ARCH}). It is mainly because the conditional quantile function
(\ref{def of q_t(theta)}) is a nonlinear function of parameters. We
verify {\bf (C4)} for the ARMA-AGARCH model in
Lemma~\ref{identification lemma} below by using the method
introduced in \citet{NohLee2013}, which may be applicable to the
models other than the ARMA-AGARCH model.}



Turning to the asymptotic normality issue of quantile regression
estimator $\hat\theta_n(\tau)$, notice that the objective function
in (\ref{def of QRE}) is not twice differentiable with respect to
$\theta$ even if $q_t(\theta)$ is smooth, and thus, a second order
Taylor's expansion is not applicable. This lack of smoothness in the
quantile regression is often overcome by using the empirical process
techniques: see, for instance, \citet{JureckovaProchazka1994} and
\citet{XiaoKoenker2009}. \citet{Huber1967} designed a method to
derive the asymptotic normality under nonstandard conditions and
\citet{Pollard1985} recast this method using the empirical process
techniques. \citet{Weiss1991}, \citet{EngleManganelli2004}, and
\citet{Komunjer2005} applied the method of \citet{Huber1967} to the
nonlinear quantile regression for
 {$\alpha$-mixing} observations.
\citet{ZhuLing2011} and \citet{LeeNoh2013} also employed the
 {method of \citet{Pollard1985}
for analyzing stationary processes.}

 {When the objective function is non-convex and non-differentiable,
it is often complicated to verify the rate of convergence of the
estimators. In this study, the root-$n$ consistency of
$\hat\theta_n(\tau)$ is proved through a local quadratic
approximation of the objective function, similarly to the one in
\citet{Pollard1985}. As a device to provide the quadratic
approximation, we derive Lemma~\ref{BR lemma} in the Appendix, which
is an extension of Lemma~3 of \citet{Huber1967} and Lemma~4 of
\citet{Pollard1985} to stationary and ergodic processes.}

In what follows, we list some additional assumptions to ensure the
asymptotic normality of $\hat\theta_n(\tau)$:

\begin{description}
    \item[(N1)]  $F_u$ has a bounded continuous density $f_u$ with $f_u(F_u^{-1}(\tau))>0$.
    \item[(N2)] $\theta^\circ(\tau)$ is an interior point of $\Theta$.
    \item[(N3)] (i) There exists a neighborhood $N_\delta$
    of $\theta^\circ(\tau)$ such that  for all $t\in\mathbb{Z}$,
         $q_t(\theta)$ is differentiable in $\theta\in N_\delta$
            and its derivative ${\partial}q_t(\theta)/{\partial \theta}$ is Lipschitz continuous  {a.s}, \\
        (ii) $E\left[ \sup_{\theta\in N_\delta} \left\| {\partial}q_1(\theta)/{\partial \theta}  \right\|^2   \right] < \infty$, \\
        (iii) $E\left[ \sup_{\theta\in N_\delta} \left\| {\partial^2}q_1(\theta)/{\partial \theta \partial\theta^T}  \right\|   \right] < \infty$.
    \item[(N4)]
        (i) For all $t\geq 1$, $\tilde{q}_t(\theta)$ is differentiable in $N_\delta$
            and its derivative  is Lipschitz continuous  {a.s}, \\
        (ii) {
$
    \sum_{t=1}^\infty \sup_{\theta\in N_\delta}  \left\| {\partial q_t(\theta)}/{\partial\theta} - {\partial\tilde{q}_t(\theta)}/{\partial\theta} \right\| < \infty
$ a.s}, \\
        (iii)  {
$
    \sum_{t=1}^\infty \sup_{\theta\in N_\delta}  \left\| { {\partial^2} {q}_t(\theta)}/{\partial \theta \partial\theta^T} -  { {\partial^2} \tilde{q}_t(\theta)}/{\partial \theta \partial\theta^T} \right\| <\infty
$  a.s.}
    \item[(N5)] Matrix $J(\tau)$ is positive definite, where
\begin{equation}\label{def of J(tau)}
        J(\tau) = E\left[ {1\over h_t({\alpha^\circ})} \frac{\partial q_t(\theta^\circ(\tau))}{\partial\theta}\frac{\partial q_t(\theta^\circ(\tau))}{\partial\theta^T}
    \right].
\end{equation}
\end{description}

\begin{remark}\label{remark: Lipschitz continuity}
  In the case of ARMA-GARCH model, $q_t(\theta)$ defined in (\ref{def of q_t(theta)}) is twice continuously
  differentiable,
  whereas the condition fails in the case of ARMA-AGARCH model:
  see Remark~\ref{Lipschitz in AGARCH} in Section~3.
  The Lipschitz continuity in {\bf (N3)} and {\bf (N4)} is intended to cover such models.
  Recall that the Lipschitz continuous functions have derivatives almost everywhere.
\end{remark}

\begin{theorem}\label{asymptotic normality}
   If the assumptions in Theorem~\ref{strong consistency} and assumptions {\bf (N1)--(N5)} hold for model~(\ref{loc-scale
   model}), then
    \begin{align*}
        \sqrt{n} ( \hat{\theta}_n(\tau) - \theta^\circ(\tau) ) \Rightarrow
        N \left(0, \frac{\tau(1-\tau)}{f_u^2(F_u^{-1}(\tau))} J(\tau)^{-1} V(\tau) J(\tau)^{-1}
        \right),
    \end{align*}
where $J(\tau)$ is given in (\ref{def of J(tau)}) and
    \begin{align*}
        V(\tau) &= E \left[ \frac{\partial q_t(\theta^\circ(\tau))}{\partial \theta} \frac{\partial q_t(\theta^\circ(\tau))}{\partial \theta^T}
        \right].
    \end{align*}
\end{theorem}

The obtained asymptotic covariance matrix
 {coincides with those for the models with location/scale components}
in \citet{JureckovaProchazka1994}, \citet{DavisDunsmuir1997},
\citet{KoenkerZhao1996}, and \citet{LeeNoh2013}. The models
considered in \citet{Weiss1991} and \citet{EngleManganelli2004}
allow a time varying conditional distribution of $Y_t$ unlike in our
study. The asymptotic covariance matrices in their results involve
the conditional density of $Y_t - Q_\tau(Y_t |\Ft)$ at $0$, which
becomes $h_t(\alpha^\circ)^{-1} f_u(F_u^{-1}(\tau))$ in our set-up.
Thus, the covariance matrix in Theorem~\ref{asymptotic normality}
can be also shown to coincide with that of
\citet{EngleManganelli2004} under the stationarity assumption.
    For the estimation of the asymptotic covariance matrix,
    we can employ the following estimator as given in \citet{Powell1991} and \citet{EngleManganelli2004}:
    \begin{align}\label{asymptotic covariance estimator}
        \tau (1-\tau) \hat{H}_n^{-1}(\tau) \hat{V}_n(\tau)  \hat{H}_n^{-1}(\tau),
    \end{align}
    where
    \begin{align*}
        \hat{V}_n(\tau) &= {1\over n} \sumtn \frac{\partial \tilde{q}_t(\hat{\theta}_n(\tau) )}{\partial \theta}
        \frac{\partial \tilde{q}_t(\hat{\theta}_n(\tau) )}{\partial \theta^T},\\
        \hat{H}_n(\tau)
        &= {1\over 2c_n n} \sumtn I\left( \left|Y_t - \tilde{q}_t(\hat\theta_n(\tau) ) \right| <c_n \right)
        \frac{\partial \tilde{q}_t(\hat{\theta}_n(\tau) )}{\partial \theta}
        \frac{\partial \tilde{q}_t(\hat{\theta}_n(\tau) )}{\partial \theta^T},
    \end{align*}
    and $c_n$ is a bandwidth satisfying $c_n\rightarrow 0$ and $\sqrt{n}c_n \to \infty$.
Theorem~3 of \citet{EngleManganelli2004} shows that the asymptotic
covariance estimator in (\ref{asymptotic covariance estimator}) is
consistent under certain regularity conditions including a more
stringent moment condition than those of Theorem~\ref{asymptotic
normality}.

\section{Quantile regression in ARMA-asymmetric GARCH models}\label{sec3}

In this section, we consider an application of the results in
Section~\ref{sec2} to the ARMA-AGARCH model taking into account
their broad usage in practice. We verify that the assumptions in
Section~\ref{sec2} hold in this model and deduce some more primitive
conditions to ensure the asymptotic properties of the quantile
regression estimator. The AGARCH model is well known to capture
asymmetric properties of conditional volatilities (see
\citealt{Glostenetal1993} and \citealt{Dingetal1993}) and to reflect
the phenomenon that past positive and negative returns impose a
different impact on current volatilities.

Let $Y_1,  \ldots, Y_n$ be  the {observations} from the
ARMA($P,Q$)-AGARCH($p,q$) model defined by
\begin{align}
    Y_t &= \phi^\circ_0 + \sum_{j=1}^P \phi^\circ_j Y_{t-j} + \sum_{i=1}^Q \psi^\circ_i \eps_{t-i} + \eps_t,    \label{ARMA model eqn}\\
    \eps_t &= h_t u_t,   \quad\quad
    h_t^2 = 1+ \sum_{i=1}^q \gamma^\circ_{1i} (\eps_{t-i}^+)^2 + \sum_{i=1}^q \gamma^\circ_{2i} (\eps_{t-i}^-)^2 + \sum_{j=1}^p \beta^\circ_j h_{t-j}^2,  \label{AGARCH model eqn}
\end{align}
where  {$a^+=\max\{a,0\}$, $a^-=\max\{-a,0\}$}, $\gamma_{li}^\circ
\geq 0$ for $l=1,2$, $i=1,\ldots,q$, $\beta_j^\circ\geq0$,
$j=1,\ldots,p$, and $\{u_t\}$ are i.i.d. random variables with
$Eu_t=0$ and $Eu_t^2=\omega^\circ$. Here, the AGARCH model in
(\ref{AGARCH model eqn}) is a reparameterized version as described
in {Section}~\ref{subsec21}. We denote by
${\varphi^\circ}=(\phi^\circ_0, \phi^\circ_1, \ldots, \phi^\circ_P,
\psi^\circ_1, \ldots, \psi^\circ_Q)^T$ and
${\vartheta^\circ}=(\gamma^\circ_{11}, \ldots, \gamma_{1q}^\circ,
\gamma_{21}^\circ, \ldots, \gamma_{2q}^\circ, \beta_1^\circ, \ldots,
\beta_p^\circ)^T$ the true ARMA and AGARCH model parameters,
respectively. Further, we denote characteristic polynomials by
$\phi^\circ(z) = 1-\sum_{j=1}^P \phi^\circ_j z^j$, $\psi^\circ(z) =
1+\sum_{i=1}^Q \psi^\circ_i z^i$, $\beta^\circ(z) =1-\sum_{j=1}^p
\beta^\circ_j z^j$, and $\gamma_l^\circ(z) = \sum_{i=1}^q
\gamma^\circ_{li} z^i$ for $l=1, 2$.

 {
The ARMA-AGARCH model (\ref{ARMA model eqn})--(\ref{AGARCH model
eqn}) admits the autoregressive representation in (\ref{loc-scale
model}) and satisfies {\bf (M1)} and {\bf(M2)} in
Section~\ref{subsec22} under some standard model assumptions: see
{\bf (A1)} and {\bf(A2)} below. \citet{Panetal2008} considered the
QML and least absolute deviation estimation for the
power-transformed and threshold GARCH models that include  AGARCH
models as a special case when the power equals 2. Theorem~5 of
\citet{Panetal2008} shows that
 equation~(\ref{AGARCH model eqn}) defines a unique strictly
stationary and ergodic solution if and only if the Lyapunov exponent
is negative: the specific formula of the exponent is given in
\citet[p.~373]{Panetal2008}. It can be seen that the Lyapunov
exponent remains the same after the reparameterization and the
condition is $E \left[ \log( \beta_1^\circ + \gamma_{11}^\circ
(u_{t}^+)^2 +\gamma^\circ_{21} (u_{t}^-)^2   ) \right]<0$ for the
AGARCH($1,1$) case. It also follows from the theorem that $\eps_t$
is a function of $\{ u_s: s\leq t\}$ and $h_t^2$ has the following
ARCH($\infty$) representation
\begin{align}\label{h_t^2 representation}
    h_t^2  = \left( 1-\sum_{j=1}^p \beta^\circ_j \right)^{-1} + \sum_{k=1}^\infty c_{1k}^\circ (\eps_{t-k}^+)^2
            + \sum_{k=1}^\infty c_{2k}^\circ (\eps_{t-k}^-)^2,
\end{align}
where $ \sum_{k=1}^\infty c_{lk}^\circ z^k =
\gamma_l^\circ(z)/\beta^\circ(z)$ for $|z|\leq 1$ and $l=1,2$. Given
the stationary AGARCH process $\{\eps_t:t\in\mathbb{Z} \}$,
assumption {\bf(A2)} below implies that $\{Y_t: t\in\mathbb{Z} \}$
is stationary and ergodic, and has the AR($\infty$) representation:
\begin{align}\label{f_t representation}
    Y_t = \left(1+\sum_{i=1}^Q \psi_i^\circ \right)^{-1}  \phi_0^\circ - \sum_{k=1}^\infty d_{k}^\circ Y_{t-k} +\eps_t,
\end{align}
where $1+\sum_{k=1}^\infty d_k^\circ z^k =
\phi^\circ(z)/\psi^\circ(z)$ for $|z|\leq 1$: see
\citet{BrockwellDavis1991}. Combining (\ref{h_t^2 representation})
and (\ref{f_t representation}), model (\ref{ARMA model
eqn})\textcolor[rgb]{1,0,0}{--}(\ref{AGARCH model eqn}) is shown to admit the autoregressive
representation in (\ref{loc-scale model}).
In addition, it follows from {\bf(A2)} that
 $Y_t$ is a function of $\{ \eps_s: s\leq t\}$,
 so is measurable with respect to the $\sigma$-field generated by $\{ u_s: s\leq
 t\}$. Therefore,
since $\mathcal{F}_t=\sigma\left( Y_s : s\leq t \right)$, {\bf (M2)}
is satisfied, and then,  the $\tau$th quantile of $Y_t$ conditional
on $\Ft$ is given by $Q_\tau(Y_t|\Ft)= f_t(\varphi^\circ) +
\xi^\circ(\tau) h_t(\varphi^\circ, \vartheta^\circ)$, where
$\xi^\circ(\tau)$ is the $\tau$th quantile of $u_1$,
$f_t(\varphi^\circ)= (1+\sum_{i=1}^Q \psi_i^\circ  )^{-1}
\phi_0^\circ - \sum_{k=1}^\infty d_{k}^\circ Y_{t-k}$, and
$h_t(\varphi^\circ, \vartheta^\circ)=h_t$ given in (\ref{h_t^2
representation}). }

 {
To estimate the conditional quantiles of $Y_t$,
we now construct the $\tau$th quantile regression estimator of
$\theta^\circ(\tau)=(\xi^\circ(\tau), {\varphi^\circ}^T,{\vartheta^\circ}^T)^T$.
Denote by $\theta=(\xi, \varphi^T, \vartheta^T)^T$ a parameter
vector which belongs to  $\Theta\subset \mathbb{R}^{P+Q+2} \times [0,\infty)^{p+2q}$.
If the parameter space $\Theta$ satisfies assumption {\bf (A4)} below,
given the stationary solution $\{Y_t: t\in\mathbb{Z}\}$ and
$\theta\in \Theta$,  we can define the stationary processes $\{
\eps_t(\varphi) \}$, $\{ h_t(\varphi, \vartheta) \}$ and
$\{q_t(\theta)\}$ consecutively as follows:
\begin{align}
    \eps_{t}(\varphi) &=   -\phi_0 + Y_t - \sum_{j=1}^P \phi_j Y_{t-j} - \sum_{i=1}^Q \psi_i \eps_{t-i}(\varphi);
    \label{def of eps_t(varphi)} \\
    h_t^2(\varphi, \vartheta) &= 1+ \sum_{i=1}^q \gamma_{1i} (\eps_{t-i}^+(\varphi))^2 + \sum_{i=1}^q \gamma_{2i} (\eps_{t-i}^-(\varphi))^2 + \sum_{j=1}^p \beta_j h_{t-j}^2(\varphi, \vartheta);  \label{def of h_t(vpt)}  \\
    q_t(\theta) &= \phi_0 + \sum_{j=1}^P \phi_j Y_{t-j} + \sum_{i=1}^Q \psi_i \eps_{t-i}(\varphi)
    + \xi h_t(\varphi, \vartheta) \label{def of AGARCH q_t(theta)}
\end{align}
for $t\in\mathbb{Z}$. Then, it can be seen that
$Q_\tau(Y_t|\Ft)=q_t(\theta^\circ(\tau))$. In practice,
$q_t(\theta)$ ($1\leq t\leq n$) cannot be computed
excepting the AR($P$)-asymmetric ARCH($q$) model case as mentioned
in Section~\ref{subsec22}. To compute an approximated conditional
quantile function, we define $\{ \tilde\eps_t(\varphi):t\geq1 \}$,
$\{ \tilde{h}_t(\varphi, \vartheta):t\geq1 \}$ and
$\{\tilde{q}_t(\theta):t\geq1\}$ by using the same equations
(\ref{def of eps_t(varphi)})--(\ref{def of AGARCH q_t(theta)}) for
$t\geq1$ and by setting the initial values
$\tilde\eps_t(\varphi)=0$, $Y_t=\phi(1)^{-1} \phi_0 $, and
$\tilde{h}_t^2(\vpt)=\beta(1)^{-1}$ for $t\leq0$. Here, we denote
$\phi(z)=1-\sum_{j=1}^P \phi_j z^j$ and $\beta(z) = \sum_{j=1}^p
\beta_j z^j$. Then, the $\tau$th quantile regression estimator
$\hat\theta_n(\tau)$  of $\theta^\circ(\tau)$ for the ARMA-AGARCH
model~(\ref{ARMA model eqn})\textcolor[rgb]{1,0,0}{--}(\ref{AGARCH model eqn}) is defined by
(\ref{def of QRE}). }

 {
To show the identifiability of the conditional quantile functions,
we introduce the following assumptions. Assumptions {\bf (A3)}(i)
and (ii) are the standard identifiability conditions for  AGARCH and
ARMA models, respectively. {\bf(A5)} assumes that $u_t$ is a
continuous random variable, which is common in real applications. }

\begin{description}
    \item[(A1)]  { $E|u_t|^\delta<\infty$ for some $\delta>0$ and the Lyapunov exponent associated with $\vartheta^\circ$ and $\{u_t\}$ is strictly negative.}
    \item[(A2)] All zeros of  $\phi^\circ(z)$ and $\psi^\circ(z)$ lie outside the unit disc.
    \item[(A3)] (i) $\gamma_1^\circ(1)+\gamma_2^\circ(1)>0$ and for each $l=1, 2$, $\gamma_{l}^\circ(z)$, $\beta^\circ(z)$ have no common zeros and
                $(\gamma_{lq}^\circ, \beta_p^\circ) \ne (0,0)$; \\
                (ii) $\phi^\circ(z)$ and $\psi^\circ(z)$ have no common zeros and $(\phi_P^\circ, \psi_Q^\circ) \ne (0,0)$.
    \item[(A4)] $\theta^\circ(\tau)\in\Theta $ and
                for all $\theta\in\Theta$, $\psi(z)=1+\sum_{i=1}^Q \psi_i z^i \ne 0$ for $|z|\leq1$
                and $\sum_{j=1}^p \beta_j <1$.
    \item[(A5)] The support of the distribution of $u_1$ is $\mathbb{R}$.
    \item[(A6)]  $E|\eps_t|<\infty$.

\end{description}

\begin{lemma}\label{identification lemma}
Suppose that  assumptions {\bf(A1)--(A5)} hold in the model~(\ref{ARMA
model eqn})--(\ref{AGARCH model eqn}) and
$q_t(\theta)=q_t(\theta^\circ(\tau))$ a.s. for some $t\in\mathbb{Z}$
and $\theta\in\Theta$. Then, we have the following:
\begin{itemize}
\item[(i)] If $\xi^\circ(\tau)\ne0$, then $\theta=\theta^\circ(\tau)$.
\item[(ii)] If $\xi^\circ(\tau)=0$, then it holds either that $\varphi=\varphi^\circ$ and $\xi=0$
or that $\phi_j=\phi_j^\circ, 1\leq j\leq P$, $\psi_i=\psi_i^\circ, 1\leq i\leq Q$, $\gamma_{1i}=\gamma_{2i}=0, 1\leq i\leq q$, and $\phi_0 + \xi \psi(1) \beta(1)^{-1/2} = \phi_0^\circ$.
\end{itemize}
\end{lemma}

 {
Lemma~\ref{identification lemma} ensures that the identifiability assumption {\bf (C4)}
for the ARMA-AGARCH model holds if $\xi^\circ(\tau)\ne0$ and it shows that only AR and MA coefficients are identifiable
in the case of $\xi^\circ(\tau)=0$.
For the consistency of $\hat\theta_n(\tau)$, we added the finite
first moment condition of the AGARCH process, which is equivalent to
$E|Y_t|<\infty$ under {\bf(A2)}. An application of
Theorem~\ref{strong consistency} and Lemma~\ref{identification
lemma} yields the strong consistency addressed below. }

\begin{remark}\label{AGARCH(1,1) moment remark}
In the GARCH case, \citet{Ling2007} presented a necessary and
sufficient condition for the stationarity and fractional moments
including {\bf(A6)}. For the AGARCH($1,1$) case, such a condition
can be obtained by using Theorem~2.1 of \citet{Ling2007} and
Theorem~6 of \citet{Panetal2008}: for $m>0$, the AGARCH($1,1$)
process $\{\eps_t\}$ is strictly stationary with
$E|\eps_t|^{2m}<\infty$ if and only if $E \left( \beta_1^\circ +
\gamma_{11}^\circ (u_{t}^+)^2 +\gamma^\circ_{21} (u_{t}^-)^2
\right)^{m} <1 $. As in there, one can use Minkowski's inequality
for $m\geq1$ and the one: $(a+b)^m\leq a^m+b^m$, $a\geq0, b\geq0$
for $0<m<1$.
\end{remark}

\begin{theorem}\label{consistency for AGARCH}
    Suppose that assumptions {\bf (C2)} and {\bf(A1)--(A6)} hold in model~(\ref{ARMA
model eqn})--(\ref{AGARCH model eqn}). Then, we have the following:
\begin{itemize}
    \item[(i)] If  $\xi^\circ(\tau)\ne0$,  $\hat\theta_n(\tau) \to \theta^\circ(\tau)$ a.s. as
    $n\to\infty$.
    \item[(ii)] If  $\xi^\circ(\tau) = 0$,
    $(\hat\phi_{1n}(\tau), \ldots, \hat\phi_{Pn}(\tau), \hat\psi_{1n}(\tau), \ldots, \hat\psi_{Qn}(\tau)  )
    \to   (\phi^\circ_1, \ldots, \phi^\circ_P, \psi^\circ_1, \ldots, \psi^\circ_Q )$ a.s. as $n\to\infty$.
\end{itemize}

\end{theorem}

To ensure the $\sqrt{n}$-consistency of $\hat\theta_n(\tau)$, moment
conditions {\bf (N3)}(ii) and (iii) are necessary. It turns out that
these conditions are implied by $EY_t^2<\infty$, or equivalently,
 $E\eps_t^2<\infty$. For the asymptotic normality, we assume the
following moment condition:

\begin{description}
    \item[(A1')] $E u_t^2<\infty$ and  $E (u_t^+)^2 \sum_{i=1}^q \gamma_{1i}^\circ  + E (u_t^-)^2 \sum_{i=1}^q \gamma_{2i}^\circ + \sum_{j=1}^p \beta_j^\circ <1$.
\end{description}

By Theorem~6.(ii) of \citet{Panetal2008}, {\bf (A1')} implies that
the model~(\ref{AGARCH model eqn}) has a stationary solution with $E
\eps_t^2 <\infty$.
Thus, {\bf (A1)} becomes redundant. Lemma~\ref{positive
definiteness} below ensures  assumption {\bf (N5)}, which is related
to the non-singularity of the asymptotic covariance matrix. The
proof of Lemma~\ref{positive definiteness} is deferred to the
supplementary material.

\begin{lemma}\label{positive definiteness}
    If assumptions {\bf(N2)}, {\bf (A1')}, and {\bf(A2)--(A5)} hold in the model~(\ref{ARMA model eqn})--(\ref{AGARCH model
    eqn})  {and $\xi^\circ(\tau)\ne0$,}
    then  $J(\tau)$ in (\ref{def of J(tau)}) and $V(\tau)$ in Theorem~\ref{asymptotic normality} are positive definite.
\end{lemma}

\begin{remark}\label{identifiability proof remark}
      Lemmas~\ref{identification lemma} and \ref{positive definiteness}
    can be verified by using a technique in \citet{NohLee2013}.
    The method shares a common idea with that used for the verification of identifiability
    in \citet{StraumannMikosch2006} and \citet{LeeLee2012},
    but is seemingly more widely applicable.
\end{remark}

\begin{theorem}\label{Asymptotics for AGARCH}
    Suppose that  assumptions {\bf(C2)}, {\bf (N1)}, {\bf(N2)}, {\bf (A1')}, and {\bf(A2)--(A5)} hold
    in the model~(\ref{ARMA model eqn})--(\ref{AGARCH model eqn}). If  {$\xi^\circ(\tau)\ne0$}, then
    $\sqrt{n} ( \hat{\theta}_n(\tau) - \theta^\circ(\tau) )$ converges in distribution to the one in Theorem~\ref{asymptotic normality}.
\end{theorem}

\begin{remark}\label{Lipschitz in AGARCH}
     {In view of (\ref{def of h_t(vpt)}) and (\ref{def of AGARCH q_t(theta)}),
    it can be shown that $\partial q_t(\theta)/\partial\theta$ is Lipschitz continuous but
    $\partial^2 q_t(\theta)/\partial\varphi\partial\varphi^T $ is discontinuous: see the proof of Theorem~\ref{Asymptotics for AGARCH}.
    In the pure AGARCH and ARMA-GARCH model cases,  $q_t(\theta)$ is twice continuously
    differentiable.}
\end{remark}

It is notable that the quantile regression yields a
$\sqrt{n}$-consistent estimation of ARMA-AGARCH parameters under the
mild moment condition of {\bf (A1')}, which is a finite second
moment condition on both the innovations and observations.
It is well known in the GARCH model that the popular Gaussian QMLE
is $\sqrt{n}$-consistent under $Eu_t^4<\infty$ but converges at a
slower rate if the innovation is heavy-tailed, that is,
$Eu_t^4=\infty$: see \citet{HallYao2003}. This fact also holds in
the reparameterized GARCH model as in Section~\ref{subsec21}: see
Section~5 of \citet{Fanetal2014}. In fact, the fourth moment
condition  of innovations is indispensable for obtaining the usual
$\sqrt{n}$-rate in various GARCH-type models: see
\citet{StraumannMikosch2006} and \citet{Panetal2008}. Further, for
mean-variance models such as the ARMA-GARCH model, the Gaussian QML
estimation additionally requires a finite fourth moment of
observations, that is, $E Y_t^4<\infty$: see
\citet{FrancqZakoian2004} and \citet{BardetWintenberger2009}.

 {In the estimation of GARCH-type models, researchers have paid considerable attention to
relaxing moment conditions and seeking robust methods against
heavy-tailed distributions of innovations or observations. For
example, \citet{BerkesHorvath2004} showed that the
$\sqrt{n}$-consistency of the two-sided exponential QMLE requires
only $Eu_t^2<\infty$ in the GARCH model, and \citet{ZhuLing2011}
verified it under $EY_t^3<\infty$ in the ARMA-GARCH model. These
moment conditions can be additionally relaxed by using weighted
likelihoods (\citealt{ZhuLing2011}) or other non-Gaussian
likelihoods (\citealt{BerkesHorvath2004, Fanetal2014}). In view of
these results, it can be reasoned that quantile regression approach
in this study also makes a reasonably good robust method in a broad
class of time series models. }

\begin{remark}
      { As mentioned in Section~\ref{subsec23}, the
    quantile regression for the location-scale models in (\ref{loc-scale model})
     requires a different conditional quantile specification when $\xi^\circ(\tau)=0$.
     Thus,
    it is necessary to test whether $\xi^\circ(\tau)$ is $0$ or not, especially for the values of $\tau$ around 0.5:
    if $\xi^\circ(\tau)=0$, the conditional quantile of $Y_t$ is just the conditional location $f_t(\alpha^\circ)$
    and the results of \citet{Weiss1991} can be applied.
    Under the null hypothesis of this testing problem, we can see that the other parameters are not identified by Lemma~\ref{identification lemma}.
    Inference in a similar situation can be found in \citet{FrancqHorvathZakoian2010} and references therein.
    We leave the development of such a test as a task of our future study. }
\end{remark}

\section{Simulation results}\label{sec4}

In this simulation study, we examine a finite sample performance of
the quantile regression estimation and illustrate its robustness
against the heavy-tailed distribution of innovations. The samples
are generated from the following ARMA($1,1$)-AGARCH($1,1$) model:
\begin{align*}
\begin{aligned}
    Y_t &= \phi_0 + \phi_1 Y_{t-1} + \psi_1 \eps_{t-1} + \eps_t, \\
    \eps_t &= h_t u_t, \quad h_t^2 = 1 + \gamma_{11} (\eps_{t-1}^+)^2 + \gamma_{21} (\eps_{t-1}^-)^2 + \beta_1 h_{t-1}^2
\end{aligned}
\end{align*}
with $Eu_t=0$, $E u_t^2=\omega$ and $( \phi_0, \phi_1, \psi_1,
\gamma_{11}, \gamma_{21}, \beta_1, \omega) = (0.04, 0.2, 0.1, 0.5,
1.25, 0.7, 0.2)$. As for the distribution of  innovation
$\omega^{-1/2} u_t$, we consider the two cases:
\begin{itemize}
    \item[(a)] standard normal distribution;
    \item[(b)] standardized skewed $t$-distribution with $4$ degrees of freedom and $0.71$ skew parameter.
\end{itemize}
The skewness of distribution (b) is approximately $-2$: see
\citet{FernandezSteel1998}. By using Remark~\ref{AGARCH(1,1) moment
remark}, we can check the stationarity and moment condition of
$\eps_t$ for the two distributions. For case (a), the AGARCH($1,1$)
process has a finite forth moment since $E\left( \beta_1+\gamma_{11}
(u_t^+)^2 + \gamma_{21} (u_t^-)^2 \right)^2 \approx 0.84<1$. For
case (b), it only holds that $Eu_t^2<\infty$ and $E\eps_t^2<\infty$
since $E\left( \beta_1+\gamma_{11} (u_t^+)^2 + \gamma_{21} (u_t^-)^2
\right)^2 \approx 2.45 > 1$ and $E\left( \beta_1+\gamma_{11}
(u_t^+)^2 + \gamma_{21} (u_t^-)^2 \right)  \approx 0.90 < 1$
according to a Monte Carlo computation.

The sample size $n$ is 2,000 and the repetition number is always
$1,000$. In computing quantile regression estimates, the Nelder-Mead
method in \verb"R" is employed and  the Gaussian-QML estimates are
used as initial values for the optimization process.

\begin{table}[ht]
\begin{center}
\caption{Performance of the quantile regression estimators for (a) $N(0,1)$}\label{tbl:Normal}
\vspace{1ex}
\begin{tabular}{ll|c c c c c c c } \hline
\multicolumn{2}{l}{}&\multicolumn{1}{c}{$\xi(\tau)$}&\multicolumn{1}{c}{$\phi_0$}&\multicolumn{1}{c}{$\phi_1$}&\multicolumn{1}{c}{$\psi_1$}
&\multicolumn{1}{c}{$\gamma_{11}$}&\multicolumn{1}{c}{$\gamma_{21}$}&\multicolumn{1}{c}{$\beta_1$}\\
\hline
$\tau=0.05$&Bias&-0.008&~0.035&~0.002&~0.010&~0.193&~0.366&-0.017\\
&SD&~0.230&~0.457&~0.238&~0.199&~0.583&~1.543&~0.097\\
&ASD&~0.262&~0.494&~0.227&~0.197&~0.325&~1.095&~0.078\\
\hline
$\tau=0.25$&Bias&-0.037&~0.086&~0.000&~0.005&~0.402&~0.398&-0.042\\
&SD&~0.238&~0.492&~0.170&~0.141&~0.982&~1.557&~0.166\\
&ASD&~0.420&~0.833&~0.172&~0.143&~0.848&~3.181&~0.123\\
\hline
$\tau=0.75$&Bias&~0.063&-0.122&-0.013&~0.015&~0.188&~0.688&-0.032\\
&SD&~0.275&~0.478&~0.160&~0.134&~0.899&~1.438&~0.146\\
&ASD&~0.350&~0.609&~0.175&~0.145&~1.098&~1.702&~0.123\\
\hline
$\tau=0.95$&Bias&-0.003&-0.008&~0.002&~0.012&~0.172&~0.407&-0.015\\
&SD&~0.257&~0.385&~0.219&~0.186&~0.823&~0.912&~0.087\\
&ASD&~0.296&~0.428&~0.243&~0.205&~0.521&~0.567&~0.077\\
\hline
\end{tabular}
\vspace{3mm}
\end{center}
\end{table}

Tables~\ref{tbl:Normal} and \ref{tbl:t4} exhibit the empirical
biases and standard deviations (SD) of the quantile regression
estimates at $\tau\in \{ 0.05, 0.25, 0.75, 0.95\}$ for cases (a) and
(b), respectively. We also report the asymptotic standard deviations
(ASD) derived from Theorem~\ref{asymptotic normality} by using  the
true parameter values and $f_u(F^{-1}_u(\tau))$. It is remarkable
that AGARCH parameters are estimated more accurately at the tail
part ($\tau=0.05, 0.95$) than at the middle part ($\tau=0.25,
0.75$). Tables~\ref{tbl:Normal} and \ref{tbl:t4} suggest that the
quantile regression method is robust against the heavy-tailed
distribution.

\begin{table}[ht!]
\begin{center}
\caption{Performance of the quantile regression estimators for (b) standardized skewed $t_4$}\label{tbl:t4}
\vspace{1ex}
\begin{tabular}{ll|c c c c c c c } \hline
\multicolumn{2}{l}{}&\multicolumn{1}{c}{$\xi(\tau)$}&\multicolumn{1}{c}{$\phi_0$}&\multicolumn{1}{c}{$\phi_1$}&\multicolumn{1}{c}{$\psi_1$}
&\multicolumn{1}{c}{$\gamma_{11}$}&\multicolumn{1}{c}{$\gamma_{21}$}&\multicolumn{1}{c}{$\beta_1$}\\
\hline
$\tau=0.05$&Bias&-0.024&~0.075&~0.010&~0.023&~0.572&~0.514&-0.047\\
&SD&~0.349&~0.755&~0.393&~0.338&~1.120&~2.001&~0.148\\
&ASD&~0.432&~0.862&~0.489&~0.421&~0.686&~1.679&~0.111\\
\hline
$\tau=0.25$&Bias&-0.064&~0.122&-0.008&~0.009&~0.508&~0.307&-0.065\\
&SD&~0.244&~0.494&~0.195&~0.174&~1.159&~1.529&~0.201\\
&ASD&~0.701&~1.381&~0.213&~0.176&~2.211&~7.147&~0.159\\
\hline
$\tau=0.75$&Bias&~0.031&-0.073&-0.004&~0.006&~0.110&~0.546&-0.008\\
&SD&~0.181&~0.359&~0.125&~0.108&~0.823&~1.357&~0.080\\
&ASD&~0.205&~0.370&~0.128&~0.107&~0.772&~1.221&~0.073\\
\hline
$\tau=0.95$&Bias&-0.012&~0.007&~0.002&~0.007&~0.122&~0.418&-0.008\\
&SD&~0.210&~0.358&~0.212&~0.188&~0.838&~0.976&~0.073\\
&ASD&~0.255&~0.425&~0.249&~0.209&~0.632&~0.664&~0.070\\
\hline
\end{tabular}
\vspace{3mm}
\end{center}
\end{table}

\begin{table}[ht!]
\begin{center}
\caption{The RMSE ratio of the Gaussian-QMLE to the quantile regression estimators}\label{tbl:RMSE ratio}
\vspace{1ex}
\begin{tabular}{lr|c c c c c c } \hline
\multicolumn{2}{l}{}&\multicolumn{1}{c}{$\phi_0$}&\multicolumn{1}{c}{$\phi_1$}&\multicolumn{1}{c}{$\psi_1$}
&\multicolumn{1}{c}{$\gamma_{11}$}&\multicolumn{1}{c}{$\gamma_{21}$}&\multicolumn{1}{c}{$\beta_1$}\\
\hline
Normal &$\tau=0.05$&0.062&0.324&0.398&0.235&0.158&0.437\\
&0.25&0.057&0.452&0.561&0.136&0.156&0.252\\
&0.75&0.058&0.480&0.588&0.157&0.157&0.289\\
&0.95&0.074&0.353&0.425&0.172&0.251&0.488\\
\hline
Skewed $t_4$ &$\tau=0.05$&0.040&0.297&0.347&0.370&0.506&0.556\\
&0.25&0.060&0.597&0.677&0.367&0.671&0.409\\
&0.75&0.083&0.935&1.091&0.560&0.715&1.078\\
&0.95&0.085&0.550&0.628&0.549&0.986&1.182\\
\hline
\end{tabular}
\vspace{3mm}
\end{center}
\end{table}

We demonstrate this robust feature in comparison with Gaussian-QMLE.
To do so, we calculate the relative efficiency defined as the ratio
of the root mean squared error (RMSE) of the Gaussian-QMLE to that
of the quantile regression estimates. Table~\ref{tbl:RMSE ratio}
shows that the relative efficiency increases in  the skewed
$t$-distribution case.

 {
It is noteworthy that the quantile regression method for pure
volatility models is identical to the CAViaR method of
\citet{EngleManganelli2004}: see Remark~9 of \citet{LeeNoh2013}. The
performance of CAViaR method has been reported in many empirical
studies. It would be interesting to examine the the performance of
our method for various location-scale models in VaR forecasting as
well. We leave this as a task of our future study. }

\section{A real data analysis}\label{sec5}

In this section, we showcase a real example of the quantile
regression for the AR($1$)-AGARCH($1,1$) model by using the daily
log returns (computed as 100 times the difference of the log prices)
of the Hong Kong Hang Seng Index series taken from Datastream from
January 4, 1993 to December 31, 2012, consisting of 5216
observations.

\begin{table}[ht!]
\begin{center}
\caption{Gaussian-QML estimation results based on the reparameterized AR($1$)-AGARCH($1,1$) model}\label{QML fits}
\vspace{1ex}
\begin{tabular}{l|cccccc}
\hline
            &   $\phi_0$ & $\phi_1$ &  $\gamma_{11}$   &  $\gamma_{21}$  &  $\beta_1$  & $\omega$   \\
\hline
 Estimates  & 0.0360 & 0.0476 & 1.2979 & 4.4150 & 0.9214 &0.0242    \\

 S.E.  &0.0180 &0.0124 &0.4659 &0.8776 &0.0102 &0.0065       \\
 $p$-values  &0.0454 &0.0001 &0.0053 &0.0000 &0.0000 &0.0002   \\
 \hline
\end{tabular}
\end{center}
\end{table}

\begin{figure}[ht!]
  \centerline{
 \includegraphics[width=16cm]{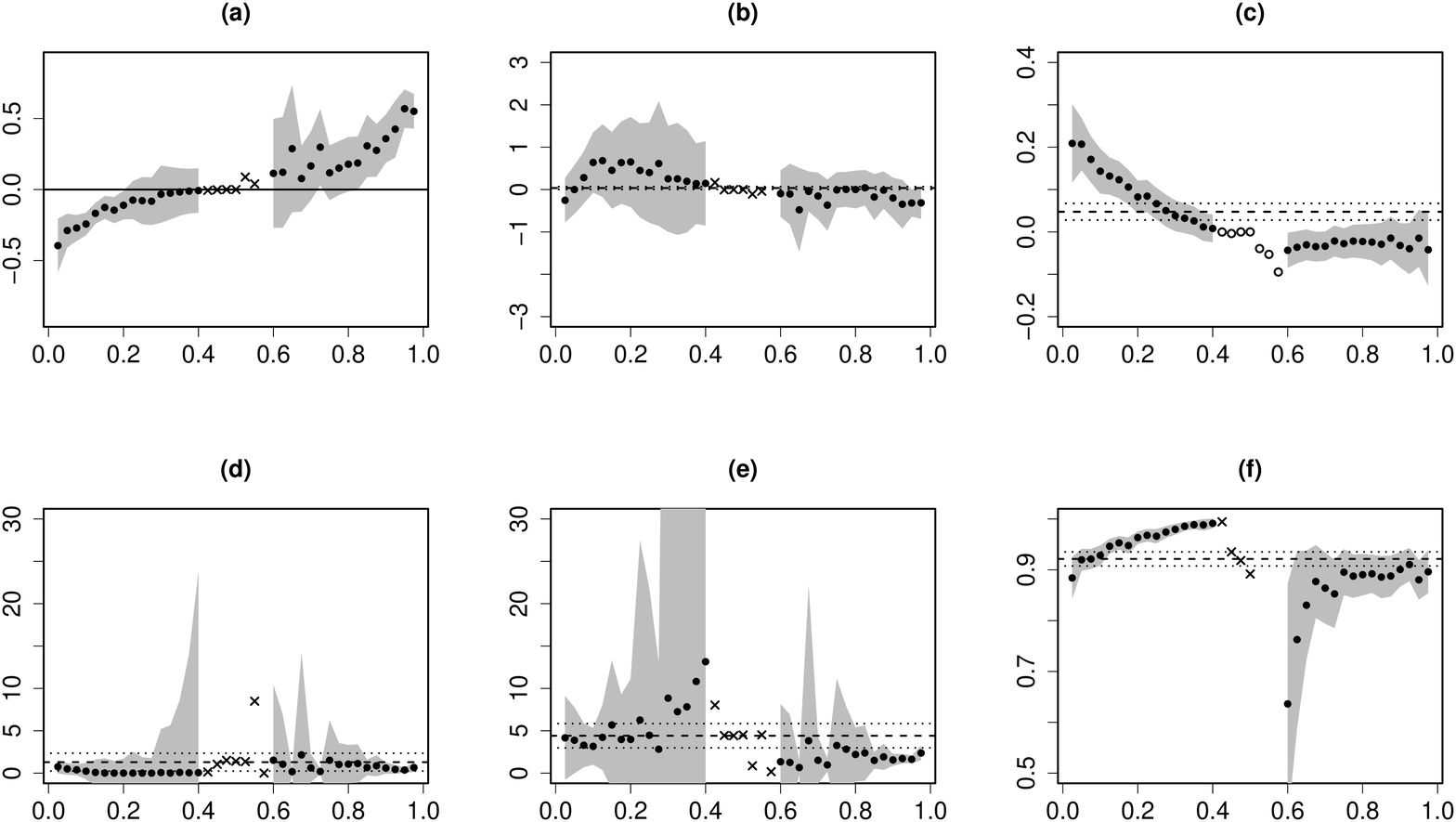}
 \vspace{0ex}
 }
\caption{Quantile regression estimates of (a) $\xi(\tau)$ (b)
$\phi_0$ (c) $\phi_1$ (d) $\gamma_{11}$ (e) $\gamma_{21}$  and (f)
$\beta_1$ at every $2.5\%$ probability level. The shaded region
illustrates $90\%$ confidence intervals. The dashed and dotted lines
represent the corresponding QML estimates and the $90\%$ confidence
intervals, respectively. For $\tau\in(0.4, 0.6)$, circles in (c)
denote consistent estimates while crosses in others denote
inconsistent ones. }\label{quantile reg fit}
\end{figure}

Table~\ref{QML fits} reports the Gaussian-QML estimates of the
parameters in model (\ref{ARMA model eqn})--(\ref{AGARCH model eqn})
with $P=1$, $Q=0$ and $p=q=1$.
The large value of $\hat\gamma_{21}$ indicates the asymmetry in
volatility, that is, negative values of returns result in a bigger
increase in future volatility than positive values. The significance
of the AR coefficient indicates that the conditional location-scale
model is better fitted to the data than pure volatility models.
Meanwhile, using the parameter estimates and residuals, it is
obtained that $n^{-1}\sumtn\{ \hat\beta_1 + \hat{\gamma}_{11}
(\hat{u}_t^+)^2 + \hat{\gamma}_{21} (\hat{u}_t^-)^2 \}\approx
0.9918$ with standard error $0.0023$, which seemingly indicates the
validity of {\bf(A1')}.

Figure~\ref{quantile reg fit} illustrates the results of the
quantile regression estimation at every $2.5\%$ probability level.
The confidence intervals are obtained based on the asymptotic
covariance estimator in (\ref{asymptotic covariance estimator}). The
test for $\xi^\circ(\tau)=0$ is not available at present, but one
can guess that $\xi^\circ(\tau)$ would be $0$ at some $\tau\in(0.4,
0.6)$ by a rule-of-thumb. Then, owing to Theorem~\ref{consistency
for AGARCH} and Lemma~\ref{identification lemma}, it can be
determined that the estimates at the $\tau$ excepting the AR
coefficients are inconsistent. Overall, our findings show that the
quantile regression estimates have the values similar to the QMLEs,
but some remarkable differences exist between both
$\hat\phi_1(\tau)$ and $\hat\phi_1^{QML}$ and $\hat\beta_1(\tau)$
and $\hat\beta_1^{QML}$ for the lower values of $\tau$. For
instance, it can be seen from (c) of Figure~\ref{quantile reg fit}
that the values of $\hat\phi_1(\tau)$ are more deviated from the
$\hat\phi_1^{QML}$ estimate in the lower conditional quantiles.
Further, it can be reasoned from (f) of Figure~\ref{quantile reg
fit} that the asymmetry of volatility still remains even after
fitting the AGARCH model.

\appendix

\section{Appendix: proofs of Theorems~\ref{strong consistency} and \ref{asymptotic normality}}\label{appendix}
\setcounter{equation}{0}
\setcounter{lemma}{0}

\renewcommand{\theequation}{A.\arabic{equation}}
\renewcommand{\thelemma}{A.\arabic{lemma}}
\renewcommand{\thecorollary}{A.\arabic{corollary}}


For simplicity, we suppress the dependence of $\theta^\circ(\tau)$
and $\xi^\circ(\tau)$ on  $\tau$. Further, we denote $\X_t=(Y_t,
Y_{t-1}, \ldots )$ and define
\begin{align*}
    g(\X_t, \theta) &= \rho_\tau(Y_t - q_t(\theta)), &
    G_n(\theta) &={1\over n} \sumtn g(\X_t, \theta), &
    \tilde{G}_n(\theta) &= {1\over n} \sum_{i=1}^n \rho_\tau(Y_t - \tilde{q}_t(\theta)),
\end{align*}
where $q_t(\theta)$ and $\tilde{q}_t(\theta)$ are those defined in
Section~\ref{subsec21}.

In the proof of the asymptotic normality, the main difficulty arises
from the lack of smoothness and stationarity of the objective
function $\tilde{G}_n(\theta)$. Lemma~\ref{G_n(theta) approximation}
below validates a quadratic approximation of $G_n(\theta)$ by
applying Lemma~\ref{BR lemma} which deals with the lack of
smoothness and extends Lemma~3 of \cite{Huber1967}. Here, we can obtain
$\hat\theta_n-\theta^\circ=O_p(n^{-1/2})$ from the approximation.
Then, Lemma~\ref{tilde G approximation} below justifies a quadratic
expansion of $\tilde{G}_n(\theta)$ in a $n^{-1/2}$-neighborhood of
$\theta^\circ$, which yields the desired asymptotic normality
result.

\ \\ \indent \textbf{Proof of Theorem~\ref{strong consistency}}. To
establish the consistency, we show that
$\tilde{G}_n(\theta)-\tilde{G}_n(\theta^\circ)$ converges uniformly
to a continuous function on $\Theta$ a.s. and the limit has a unique
minimum at $\theta^\circ$. Let $C(\Theta)$ be the space of
continuous real-valued functions on $\Theta$ equipped with the
sup-norm. Since $\{ Y_t \}_{ t\in\mathbb{Z} }$ is strictly
stationary ergodic, {\bf (C3)} implies that $\{ g(\X_t, \cdot) \}_{
t\in\mathbb{Z} }$ is a stationary ergodic sequence of
$C(\Theta)$-valued random elements: see Proposition~2.5 of \cite{StraumannMikosch2006}.
Note that due to the Lipschitz continuity of
$\rho_\tau(\cdot)$ and {\bf (C3)}(ii), we have $
    E \left[ \sup_{\theta\in\Theta} \left| g(\X_1, \theta) - g(\X_1, \theta^\circ) \right| \right]  <\infty
$. Hence, by applying the ergodic theorem (see Theorem~2.7 of \citealt{StraumannMikosch2006}), it follows that
\begin{align*}
    \sup_{\theta\in\Theta} \left| G_n(\theta)-G_n(\theta^\circ) - \Gamma(\theta) \right| \xar[]{} 0 \quad\text{a.s.},
\end{align*}
where
$
    \Gamma(\theta) = E \left[ g(\X_1, \theta) - g(\X_1, \theta^\circ) \right]
$.
Also, from {\bf (C6)}, we have
\begin{align}\label{difference of tilda G and G}
    \sup_{\theta\in\Theta} \left| \tilde{G}_n(\theta) - \tilde{G}_n(\theta^\circ) - \left\{ G_n(\theta)- G_n(\theta^\circ) \right\} \right|
    \leq 2 \sup_{\theta\in\Theta}  {1\over n} \sumtn \left| q_t(\theta) - \tilde{q}_t(\theta) \right|
    \to 0  \quad\text{a.s.}
\end{align}

Now, we show that $\Gamma(\theta)$ is uniquely minimized at
$\theta=\theta^\circ$. Recall that $Y_t = f_t({\alpha^\circ}) +
h_t({\alpha^\circ}) u_t$ and $q_t(\theta^\circ)=f_t({\alpha^\circ})
+ \xi^\circ h_t({\alpha^\circ})$. By {\bf (C5)} and the fact that
$\rho_\tau(cx)=c \rho_\tau(x)$, $c>0$, we have
\begin{align*}
    \Gamma(\theta) &= E \left[  h_t({\alpha^\circ}) \left\{ \rho_\tau \left( u_t - \frac{q_t(\theta)-f_t({\alpha^\circ})}{h_t({\alpha^\circ})} \right)
    - \rho_\tau \left( u_t - \xi^\circ \right) \right\} \right] \\
    &= E \left[ h_t({\alpha^\circ}) H\left( \frac{q_t(\theta)-f_t({\alpha^\circ})}{h_t({\alpha^\circ})} \right) \right],
\end{align*}
where $ H(x) \equiv E \left[ \rho_\tau(u_t - x ) - \rho_\tau (u_t -
\xi^\circ) \right] $. It can be easily checked that under {\bf
(C1)}, $H(x)\geq 0$ and $H(x)=0$ if and only if $x=\xi^\circ$: see
(2.9) of \cite{BassettKoenker1986}. Hence, $\Gamma(\theta) \geq 0$
and $\Gamma(\theta)=0$ if and only if $\left\{
{q_t(\theta)-f_t({\alpha^\circ})}
\right\}/{h_t({\alpha^\circ})}=\xi^\circ$ a.s. for some
$t\in\mathbb{Z}$. Since {\bf (C4)} directly indicates that
$\Gamma(\theta)$ has a unique minimum at $\theta^\circ$, the theorem
is established by a standard compactness argument. \hfill \fbox{}

\begin{lemma}\label{BR lemma}
    Let $\{ \mathcal{G}_t : t\in\mathbb{Z} \}$ be a sequence of nondecreasing $\sigma$-fields.
    Suppose that $\{Z_t : t\in\mathbb{Z} \}$ is a strictly stationary ergodic sequence of random variables
    and $Z_t$ is measurable with respect to $\mathcal{G}_t$, say $Z_t\in \mathcal{G}_t$, for all $t$.
    Define $\Z_t=(Z_t, Z_{t-1}, \ldots )$.
    For $\theta^\circ\in\mathbb{R}^d$ and $\theta$ near $\theta^\circ$, let
    $f(\cdot, \theta):\mathbb{R}^\infty \to \mathbb{R}$ be measurable functions such that $f(\cdot, \theta^\circ)=0$.
Suppose that the following conditions hold:
\begin{itemize}
  \item[(a)] $E \left[ \sup_{\|h\|\leq R} \left| f(\Z_t, \theta^\circ+h) \right|^2 \right] \to 0$ as $R\to 0$.
  \item[(b)] There exist a $d_0>0$ and a stationary ergodic sequence $\{B_t \}$ with $E[B_t]<\infty$ and $B_t\in\mathcal{G}_{t}$  such that
for all $t$  and $R\geq0$,
\begin{align*}
    \sup_{\theta:\|\theta-\theta^\circ \|+R\leq d_0 } E\left[ \sup_{\|h\|\leq R} \left| f(\Z_t, \theta+h) - f(\Z_t, \theta) \right| \Big| \mathcal{G}_{t-1} \right] \leq B_{t-1} R.
\end{align*}
\end{itemize}
Then, as $n\to\infty$,
\begin{align}\label{Stochastic Differentiability Condition}
    \sup_{\|\theta-\theta^\circ\|\leq d_0}  \frac{\left| \mathcal{W}_n(f(\cdot,\theta)) \right|}{1+\sqrt{n}\|\theta-\theta^\circ\|} \topr 0,
\end{align}
where
$\mathcal{W}_n(f(\cdot,\theta)) = n^{-1/2} \sumtn \left\{ f(\Z_t, \theta) - E\left[ f(\Z_t, \theta) | \mathcal{G}_{t-1} \right] \right\}
$.
\end{lemma}

\begin{proof}
  The proof is essentially the same as that of Lemma~4 of \cite{Pollard1985} except that
the summands in $\mathcal{W}_n(f(\cdot,\theta))$ are not i.i.d. but
a sequence of martingale differences. We take $d_0$ to be 1 for
convenience. First, we show that (b) implies that $\mathfrak{F}:=\{
f(\cdot,\theta) : \|\theta-\theta^\circ\|\leq1 \}$ satisfies the
bracketing condition in \cite{Pollard1985}. Denote $b_0=E[B_t]$.
Given $\eps>0$ and $0< R \leq 1$, there exist open balls
$B(\theta_i, (2b_0)^{-1}\eps R)$, $i=1, 2, \ldots, K_\eps$, covering
$B(\theta^\circ, R)$. Notice that the same $K_\eps$ works for every
$R$. Thus, there is a partition $\{ U_i(R) : i=1, 2, \ldots, K_\eps
\}$ of $B(\theta^\circ, R)$
 such that $U_i(R)\subset B(\theta_i,(2b_0)^{-1}\eps R)$.
For each partition, upper and lower bracketing functions
$\overset{-}{f}_i$, $\overset{\circ}{f}_i$ are defined as $f(\Z_t,
\theta_i) \pm \sup_{\theta\in U_i(R)} |f(\Z_t, \theta) - f(\Z_t,
\theta_i)|$, respectively. Using condition (b), it follows that
\begin{align}\label{bound of E_t-t of BR ftns}
\begin{aligned}
    E\left[ \overset{-}{f}_i(\Z_t) - \overset{\circ}{f}_i(\Z_t) \Big| \mathcal{G}_{t-1} \right]
    &\leq  2 E\left[  \sup_{\|h\|\leq (2b_0)^{-1}\eps R} \left| f(\Z_t, \theta_i +h) - f(\Z_t, \theta_i) \right| \Big| \mathcal{G}_{t-1} \right] \\
    &\leq  {B_{t-1} b_0^{-1} \eps R  },
\end{aligned}
\end{align}
so that $E\left[ \overset{-}{f}_i(\Z_t) - \overset{\circ}{f}_i(\Z_t)  \right] \leq \eps R$.
Hence, $\mathfrak{F}$ satisfies the bracketing condition.


For each $k \in\{0\}\cup \mathbb{N}$, put $R(k) ={ 2^{-k} }$. Let
$B(k)$ be the ball of radius $R(k)$ centered at $\theta^\circ$ and
let $A(k)$ be the annulus $B(k)-B(k+1)$. Then, for given $\eps>0$
and $k$, there is a partition $U_1(R(k)), ~U_2(R(k)), \ldots,
~U_{K_\eps}(R(k))$ of $B(k)$. It follows from (\ref{bound of E_t-t
of BR ftns}) that for $\theta\in U_i(R(k))$,
\begin{align*}
    \mathcal{W}_n(f(\cdot,\theta)) &\leq {1\over\sqrt{n}} \sumtn \left\{ \overset{-}{f}_i (\Z_t) - E\left[ \overset{-}{f}_i(\Z_t) \Big| \mathcal{G}_{t-1} \right]     \right\}
    + {1\over\sqrt{n}} \sumtn   E\left[ \overset{-}{f}_i(\Z_t) -  \overset{\circ}{f}_i(\Z_t) \Big| \mathcal{G}_{t-1} \right]  \\
    &\leq \mathcal{W}_n( \overset{-}{f}_i(\cdot) ) + \sqrt{n} \eps R(k) \left( {1\over n b_0}\sumtn {B_{t-1}} \right).
\end{align*}
If we set $C_n=\left( (n b_0)^{-1} \sumtn B_{t-1} \leq 2 \right)$,
$P(C_n)$ tends to 1 by ergodicity. Further, as in \cite{Pollard1985}, it
can be seen that
\begin{align*}
    P\left( \sup_{A(k)} \frac{ \mathcal{W}_n(f(\cdot,\theta)) }{1+\sqrt{n}\|\theta-\theta^\circ \|} > 8\eps, C_n \right)
    \leq K_\eps \max_{1\leq i \leq K_\eps} P\left( \mathcal{W}_n(\overset{-}{f}_i(\cdot) ) > 2\eps \sqrt{n} R(k)
    \right).
\end{align*}
Then, using the arguments as in the rest part of the proof of
Lemma~4 of \cite{Pollard1985}, we can establish the lemma.
\end{proof}

\begin{lemma}\label{G_n(theta) approximation}
    Under assumptions {\bf (C3), (C5)} and {\bf (N1)--(N3)}, we have
\begin{multline*}
    G_n(\theta)- G_n(\theta^\circ)
    = {{f_u(\xi^\circ) \over2} } (\theta-\theta^\circ)^T J(\tau) (\theta-\theta^\circ) \\
    + {n^{-1/2} (\theta-\theta^\circ)^T } \left[ n^{-1/2} \sumtn     \frac{\partial q_t(\theta^\circ)}{\partial\theta} \left\{ I(Y_t<q_t(\theta^\circ))-\tau \right\} \right]
    + {n^{-1/2}}{\|\theta-\theta^\circ \|}  R_n(\theta),
\end{multline*}
where $J(\tau)$ is defined in (\ref{def of J(tau)}) and as
$n\to\infty$,
\begin{align*}
    \sup_{\|\theta-\theta^\circ \|\leq r_n} \frac{|R_n(\theta)|}{1+\sqrt{n}\|\theta-\theta^\circ \|} \topr 0
\end{align*}
 for every sequence $\{r_n\}_{n\in\mathbb{N}}$ tending to 0.
\end{lemma}

\begin{proof}
Note that $\rho_\tau(x)$ is Lipschitz continuous in $x$ and its
derivative is $\psi_\tau(x) \equiv \tau-I(x<0)$ excepting $x=0$. By
{\bf (N3)}(i), $g(\X_{t},\theta)$ is Lipschitz continuous in
$\theta\in N_\delta$ with probability 1. Thus,
$g(\X_{t},\theta^\circ + u(\theta-\theta^\circ) )$ is absolutely
continuous in $u\in [0,1]$, so is differentiable at every $u$
outside a set of Lebesgue measure 0. Hence, by the fundamental
theorem of  calculus, we can express
\begin{align*}
    g(\X_{t},\theta) - g(\X_{t},\theta^\circ) &= (\theta-\theta^\circ)^T \int_0^1  g_1(\X_t,\theta^\circ + u(\theta-\theta^\circ) ) \md u,
\end{align*}
where $g_1(\X_t,\theta) \equiv    \left\{ I(Y_t < q_t(\theta)) -
\tau \right\}  {\partial q_t(\theta)}/{\partial\theta} $. This with
{\bf (N3)}(ii) yields
\begin{align}\label{E g theta - g theta_0}
    E\left[ g(\X_t,\theta) - g(\X_t,\theta^\circ) |\Ft \right]
    &= (\theta-\theta^\circ)^T \int_0^1 E\left[ g_1(\X_t, \theta(u) ) |\Ft \right] \md u,
\end{align}
where $\theta(u)\equiv\theta^\circ + u(\theta-\theta^\circ)$ Then,
following the arguments as in \cite{Pollard1985}, we get
\begin{equation}\label{decomposition}
\begin{aligned}
    & G_n(\theta)-G_n(\theta^\circ) \\
    &={1\over n} \sumtn E\left[ g(\X_t,\theta) - g(\X_t,\theta^\circ) |\Ft \right]
    +  {(\theta-\theta^\circ)^T \over n} \sumtn  g_1(\X_t, \theta^\circ )
    + {(\theta-\theta^\circ)^T \over\sqrt{n}} R_{1n}(\theta),
\end{aligned}
\end{equation}
where $$
    R_{1n}(\theta):
    = \int_0^1 {1\over\sqrt{n}} \sumtn \left\{   g_1(\X_t, \theta(u) ) -  g_1(\X_t, \theta^\circ)
    - E\left[  g_1(\X_t, \theta(u) ) | \Ft \right]
    \right\} \md u$$ satisfies
\begin{equation}\label{for R_1n}
\begin{aligned}
    &\sup_{\|\theta-\theta^\circ \|\leq r_n} \frac{\|R_{1n}(\theta)\|}{1+\sqrt{n}\|\theta-\theta^\circ \|} \\
    &\leq
    \sup_{\|\theta-\theta^\circ \|\leq r_n} \int_0^1 \frac{ \left\| {1\over\sqrt{n}} \sumtn \left\{   g_1(\X_t, \theta(u) ) -  g_1(\X_t, \theta^\circ)
    - E\left[  g_1(\X_t, \theta(u) ) |\Ft \right]
    \right\}  \right\| }{1+\sqrt{n}\|\theta(u)-\theta^\circ \|} \md u \\
    &\leq \sup_{\|\theta-\theta^\circ \|\leq r_n} \frac{\| \mathcal{W}_n(r(\cdot, \theta))\|}{1+\sqrt{n}\|\theta-\theta^\circ \|},
\end{aligned}
\end{equation}
where $r(\X_t, \theta) \equiv I(Y_t < q_t(\theta)) { \partial
q_t(\theta) / {\partial\theta} } -  I(Y_t < q_t(\theta^\circ)) {
\partial q_t(\theta^\circ)  / {\partial\theta} }$ and
$\mathcal{W}_n(\cdot)$ is defined in (\ref{Stochastic
Differentiability Condition}). Define $e_t(\theta) = E\left[
g(\X_t,\theta) - g(\X_t,\theta^\circ) |\Ft \right]$. In view of
(\ref{decomposition}) and (\ref{for R_1n}), it suffices to verify
that for every sequence of $\{r_n\}_{n\in\mathbb{N}}$ tending to 0,
\begin{align}\label{StDiffCond}
    \sup_{\|\theta-\theta^\circ \|\leq r_n} \frac{\| \mathcal{W}_n(r(\cdot, \theta))\|} { {1+\sqrt{n}\|\theta-\theta^\circ \|} } &\topr 0
\end{align}
and
\begin{align}\label{Calculation of cond mean}
    \sup_{\|\theta-\theta^\circ \|\leq r_n} \left| {1\over n} \sumtn e_t(\theta)
    - {f_u(\xi^\circ)\over2}(\theta-\theta^\circ)^T J(\tau) (\theta-\theta^\circ) \right| / \|\theta-\theta^\circ \|^2 &\topr
    0.
\end{align}

We first verify (\ref{StDiffCond}) utilizing Lemma~\ref{BR lemma}.
Note that for $\|\theta-\theta^\circ \|\leq R \leq \delta$,
\begin{align*}
    I\left(Y_t < q_t(\theta^\circ) - R M_{1t} \right) \leq I(Y_t <q _t(\theta)) \leq I\left(Y_t < q_t(\theta^\circ) + R M_{1t} \right),
\end{align*}
where $M_{1t} \equiv \sup_{\theta\in N_\delta} \left| {\partial q_t(\theta)}/{\partial\theta} \right|$.
Using this and the inequality $|ab-cd|\leq |a-c||b|+|c||b-d|$,
we have that for all small $R$,
\begin{align*}
    &\sup_{\|\theta-\theta^\circ \|\leq R} \left\| r(\X_t, \theta) \right\|^2
    \leq  2 \sup_{\|\theta-\theta^\circ \|\leq R} \left\| {\partial q_t(\theta) \over \partial\theta } - {\partial q_t(\theta^\circ) \over \partial\theta } \right\|^2 \\
    &+ 2    \left\| {\partial q_t(\theta^\circ) \over \partial\theta } \right\|^2 \left\{
    I\left(Y_t < q_t(\theta^\circ) + R M_{1t} \right) - I\left(Y_t < q_t(\theta^\circ) - R M_{1t} \right) \right\}
    \leq 10 M_{1t}^2.
\end{align*}
Thus, using the dominated convergence theorem,  {\bf (N1)} and {\bf
(N3)}, we can have
\begin{align}\label{checking condition (a)}
    \lim_{R\to 0} E\left[ \sup_{\|\theta-\theta^\circ \|\leq R} \left\| r(\X_t, \theta) \right\|^2 \right]
    &\leq 0+ 2 E\left[ \left\| {\partial q_t(\theta^\circ) \over \partial\theta } \right\|^2 I(Y_t=q_t(\theta^\circ)) \right] = 0.
\end{align}
Similarly, for all $\theta$ with $\|\theta-\theta^\circ \|+R \leq
\delta$ and $R\geq 0$,
\begin{align*}
    \sup_{|h|\leq R} \left\| r(\X_t, \theta +h) - r(\X_t, \theta) \right\|
    \leq R M_{2t} + M_{1t} \left\{ I(Y_t < q_t(\theta) + R M_{1t}) - I(Y_t < q_t(\theta) - R M_{1t}) \right\},
\end{align*}
where $M_{2t} \equiv \sup_{\theta\in N_\delta} \left\|
{\partial^2}q_t(\theta)/{\partial \theta \partial\theta^T}
\right\|$. As in the proof of Theorem~\ref{strong consistency}, it
can be shown that $\{M_{1t}\}$ and $\{ M_{2t} \}$ are stationary and
ergodic due to {\bf (N3)}(i). Further, $M_{1t}$ and $ M_{2t} $ are
$\Ft$-measurable for all $t\in\mathbb{Z}$. Note that by the mean
value theorem, {\bf (N1)} and {\bf (C5)},
\begin{align*}
    &\left| E\left[ I(Y_t < q_t(\theta) + R M_{1t}) - I(Y_t < q_t(\theta) - R M_{1t}) |\Ft
    \right] \right|
    \leq {2 c_0^{-1} \|f_u\|_{\infty}RM_{1t}} ,
\end{align*}
where $\|f_u\|_{\infty}=\sup_x |f_u(x)|$, so that for all $\theta$
with $\|\theta-\theta^\circ \|+R \leq \delta$ and $R\geq 0$,
\begin{align}\label{checking condition (b)}
    E\left[ \sup_{\|h\|\leq R} \left\| r(\X_t, \theta +h) - r(\X_t, \theta) \right\|  \Big| \Ft \right]
    \leq \left( M_{2t} + {2  c_0^{-1} \|f_u\|_{\infty}M_{1t}^2} \right) R.
\end{align}
Then, combining (\ref{checking condition (a)}) and (\ref{checking
condition (b)}) and applying Lemma~\ref{BR lemma} componentwise, we
get (\ref{StDiffCond}).

Next, we verify (\ref{Calculation of cond mean}). In view of (\ref{E
g theta - g theta_0}) and {\bf (N1)}, we have that for $\theta\in
N_\delta$,
\begin{align}\label{e_t theta}
    e_t(\theta) = (\theta-\theta^\circ)^T \int_0^1 \frac{\partial q_t(\theta(u))}{\partial\theta}
     \left\{ F_u\left(\xi^\circ + \frac{q_t(\theta(u))-q_t(\theta^\circ)}{h_t({\alpha^\circ})} \right) - \tau \right\} \md
     u,
\end{align}
and thus, $\partial e_t(\theta^\circ) / \partial\theta = 0$. As
mentioned in Remark~\ref{remark: Lipschitz continuity}, owing to
{\bf(N3)}(i), we can express $
   \displaystyle \frac{\partial^2 e_t(\theta)}{\partial \theta \partial \theta^T}
    = A_{1t}(\theta) + A_{2t}(\theta),
$ where
\begin{align*}
    A_{1t}(\theta) &= f_u\left( \xi^\circ + \frac{q_t(\theta)-q_t(\theta^\circ)}{h_t({\alpha^\circ})} \right)
    {1\over h_t({\alpha^\circ})} {\partial q_t(\theta) \over \partial\theta } {\partial q_t(\theta) \over \partial\theta^T}, \\
    A_{2t}(\theta) &= \left\{ F_u\left(\xi^\circ + \frac{q_t(\theta)-q_t(\theta^\circ)}{h_t({\alpha^\circ})} \right) - \tau \right\}
    \frac{\partial^2 q_t(\theta)}{\partial \theta \partial \theta^T}.
\end{align*}
Hence, by using the fundamental theorem of calculus, the term in
(\ref{Calculation of cond mean}) can be seen to be no more than
\begin{align}\label{Taylor expansion for E_n(theta)}
    &\sup_{\|\theta-\theta^\circ \|\leq r_n} \left\|
    \int_0^1 {1\over n} \sumtn \frac{\partial^2 e_t(\theta(u))}{\partial \theta \partial \theta^T}
    (1-u) \md u - {f_u(\xi^\circ)\over 2} J(\tau)
    \right\| \nonumber \\
    &\leq \sup_{\|\theta-\theta^\circ \|\leq r_n} \left\| {1\over n} \sumtn A_{1t}(\theta) - f_u(\xi^\circ) J(\tau) \right\|
    +  \sup_{\|\theta-\theta^\circ \|\leq r_n} \left\| {1\over n} \sumtn A_{2t}(\theta) \right\|.
\end{align}
As in the proof of Theorem~\ref{strong consistency}, owing to
{\bf(N3)}(i), $\{ A_{1t}(\cdot) \}_{t\in\mathbb{Z}}$ forms a
stationary and ergodic sequence of random elements with values in
the space of continuous functions from $N_\delta$ to
$\mathbb{R}^{d\times d}$. Further, by {\bf(C5)}, {\bf (N1)} and {\bf
(N3)}(ii), we have $E\left[\sup_{\theta\in N_\delta} \left\|
A_{1t}(\theta) \right\|\right]<\infty$. Thus,
$
    \sup_{\theta\in N_\delta} \left\| {1\over n}\sumtn A_{1t}(\theta)
    - E\left[ A_{1t}(\theta) \right] \right\|
    \topr 0$ by Theorem~2.7 of
\cite{StraumannMikosch2006}. Then, since $E\left[
A_{1t}(\theta^\circ) \right]=f_u(\xi^\circ) J(\tau)$, the first term
on the right-hand side of (\ref{Taylor expansion for E_n(theta)}) is
$\op(1)$.

Since the second derivative of $q_t(\theta)$ is not necessarily
continuous, we have to take an approach similar to that used to
verify Lemma~2.3 of \cite{ZhuLing2011}. Owing to {\bf(N3)}(iii),
using the dominated convergence theorem, we can have $
    \lim_{R\to0} E\left[ \sup_{\|\theta-\theta^\circ \|\leq R} \left\| A_{2t}(\theta)\right\|
    \right]=0,$
and thus, for any $\eps>0$, there exists $R>0$ such that
\begin{align*}
    P\left(  \sup_{\|\theta-\theta^\circ \|\leq r_n} \left\| {1\over n} \sumtn A_{2t}(\theta) \right\| >\eps
    \right)
    &\leq P\left( {1\over n} \sumtn \sup_{\|\theta-\theta^\circ \|\leq R} \left\| A_{2t}(\theta)\right\| >\eps
    \right)
    < \eps
\end{align*}
for all large $n$ with $r_n\leq R$. Therefore, (\ref{Calculation of cond mean}) is verified, which completes the proof.
\end{proof}

\begin{lemma}\label{tilde G approximation}
    Under the conditions in Lemma~\ref{G_n(theta) approximation} and {\bf (N4)}, we have
\begin{multline*}
    \tilde{G}_n(\theta)- \tilde{G}_n(\theta^\circ)
    = {{f_u(\xi^\circ) \over2} } (\theta-\theta^\circ)^T J(\tau) (\theta-\theta^\circ) \\
    + {n^{-1/2} (\theta-\theta^\circ)^T } \left[ n^{-1/2} \sumtn     \frac{\partial q_t(\theta^\circ)}{\partial\theta} \left\{ I(Y_t<q_t(\theta^\circ))-\tau \right\} \right]
    +  R_n(\theta),
\end{multline*}
where $ \sup_{\|\theta-\theta^\circ \|\leq C n^{-1/2}  } n|R_n(\theta)|
\topr 0$ and  $C$ is any positive real number.
\end{lemma}

The proof of Lemma~\ref{tilde G approximation} is deferred to the supplementary material.

\ \\ \indent \textbf{Proof of Theorem~\ref{asymptotic normality}}.
We first improve the rate of convergence of $\hat\theta_n$ from
$o_p(1)$ to $O_p(n^{-1/2})$ by using Lemma~\ref{G_n(theta)
approximation} and (\ref{difference of tilda G and G}) and then
establish the theorem by using Lemma~\ref{tilde G approximation}.
Since  $\hat\theta_n$ lies in a shrinking neighborhood of
$\theta^\circ$ with probability tending to 1 due to
Theorem~\ref{strong consistency}, Lemma~\ref{G_n(theta)
approximation} and (\ref{difference of tilda G and G}) yield that
\begin{align*}
    \tilde{G}_n(\hat\theta_n) - \tilde{G}_n(\theta^\circ) =
    {{f_u(\xi^\circ) \over2} } (\hat\theta_n-\theta^\circ)^T J(\tau) (\hat\theta_n-\theta^\circ)
    + {n^{-1/2} (\hat\theta_n-\theta^\circ)^T } W_n
    + R_{1,n} + R_{2,n},
\end{align*}
where $R_{1,n} = o_p(n^{-1/2}\|\hat\theta_n-\theta^\circ\| + \|\hat\theta_n-\theta^\circ\|^2)$, $R_{2,n} = O_p(n^{-1})$ and
\begin{align*}
    W_n &=  n^{-1/2} \sumtn     \frac{\partial q_t(\theta^\circ)}{\partial\theta} \left\{ I(Y_t<q_t(\theta^\circ))-\tau \right\}.
\end{align*}
As in the proof of Theorem~\ref{strong consistency}, it is easily
checked that the summands in $\{W_n\}$ is stationary and ergodic. By
using {\bf (N3)}(ii) and applying the CLT for stationary ergodic
martingale difference sequences (e.g., \citealt{Billingsley1961}),
we can show that $W_n \Rightarrow N(0, \tau(1-\tau) V(\tau))$. Then,
from {\bf (N1)}, {\bf (N5)} and the fact that
$\tilde{G}_n(\hat\theta_n) - \tilde{G}_n(\theta^\circ) \leq 0$, the
$\sqrt{n}$-consistency of $\hat\theta_n$ can be obtained by some
algebras as seen in the proof of Theorem~2 of \cite{LeeNoh2013}.

Now, we put $U_n=- \left\{ f_u(\xi^\circ) J(\tau) \right\}^{-1}
n^{-1/2} W_n$ and use Lemma~\ref{tilde G approximation} to get
\begin{align*}
    &\tilde{G}_n(\hat\theta_n) - \tilde{G}_n(\theta^\circ)
    = {1\over2} (\hat\theta_n-\theta^\circ)^T \left\{ f_u(\xi^\circ)J(\tau) \right\} (\hat\theta_n-\theta^\circ)
    - (\hat\theta_n-\theta^\circ)^T \left\{ f_u(\xi^\circ)J(\tau) \right\}  U_n +
    o_p(n^{-1}),\\
   &\tilde{G}_n(\theta^\circ + U_n) - \tilde{G}_n(\theta^\circ)
    = -{1\over2} U_n^T \left\{ f_u(\xi^\circ)J(\tau) \right\} U_n + o_p(n^{-1}).
\end{align*}
Whence, as in the proof of Theorem~2 of \cite{Pollard1985}, it can
be seen  that the inequality $\tilde{G}_n(\hat\theta_n) \leq
\tilde{G}_n(\theta^\circ+U_n)$ yields  $ n^{1/2} ( \hat{\theta}_n -
\theta^\circ ) = n^{1/2} U_n + o_p(1)$, which together with
Slutsky's lemma asserts the theorem. \hfill \fbox{}
\\

{\bf Acknowledgements}. This work was supported by the National
Research Foundation of Korea(NRF) grant funded by the Korea
government(MSIP) (No.2012R1A2A2A01046092) (S. Lee), and National
Research Foundation of Korea Grant funded by the Korean Government
(Ministry of Education, Science and Technology)
(NRF-2011-355-C00022) (J. Noh).
\\

{\bf Supplementary material}. A supplementary material contains the
proofs of Lemma~\ref{tilde G approximation},
Lemmas~\ref{identification lemma}--\ref{positive definiteness} and
Theorems~\ref{consistency for AGARCH}--\ref{Asymptotics for AGARCH}.


\newpage

\setcounter{footnote}{0}

\begin{center}
 {\bf\LARGE  Supplementary Material to
 ``Quantile Regression for Location-Scale Time Series
 Models with Conditional Heteroscedasticity"}
 \vspace{.5cm}
\end{center}
\begin{center}
Jungsik Noh\footnote{Quantitative Biomedical Research Center, Department of Clinical Sciences, University of Texas Southwestern Medical Center, Dallas, TX 75390, USA. Email: nohjssunny@gmail.com}
and
Sangyeol Lee\footnote{Department of Statistics, Seoul National
University, Seoul 151-747, Korea. Email: sylee@stats.snu.ac.kr}\\
$^1$University of Texas Southwestern Medical Center \\
$^2$Seoul National University
\end{center}

\begin{abstract}
    { The aim of this supplementary material is to provide the proofs of Lemma~\ref{tilde G approximation}, Lemmas~\ref{identification lemma}--\ref{positive definiteness},
and Theorems~\ref{consistency for AGARCH}--\ref{Asymptotics for AGARCH} used for obtaining the results stated in the main article.
    }
\end{abstract}

\section{Supplementary}
\setcounter{equation}{0}
\setcounter{lemma}{0}

\renewcommand{\theequation}{S.\arabic{equation}}

\textbf{Proof of Lemma~\ref{tilde G approximation}}. By using the
same arguments to obtain (\ref{decomposition}), we can see that due
to {\bf (N4)}, for $\theta\in N_\delta$,
\begin{align*}
    \tilde{G}_n(\theta)- \tilde{G}_n(\theta^\circ) ={n^{-1}} \sumtn \tilde{e}_t(\theta)
    +  {n^{-1} (\theta-\theta^\circ)^T } \sumtn  \tilde{g}_{1t}(\theta^\circ )
    + {n^{-1/2} (\theta-\theta^\circ)^T } \tilde{R}_{1n}(\theta),
\end{align*}
where $\tilde{e}_t(\theta)$, $\tilde{g}_{1t}(\theta)$ and
$\tilde{R}_{1n}(\theta)$ are the same as $e_t(\theta)$, $g_1(\X_t,
\theta)$ and $R_{1n}(\theta)$ in Lemma~\ref{G_n(theta)
approximation} with $q_t(\cdot)$ replaced by $\tilde{q}_t(\cdot)$,
respectively. To establish the lemma, it suffices to show that
\begin{align}
    &\left\| {n^{-1/2}} \sumtn \tilde{g}_{1t}(\theta^\circ) - {n^{-1/2}}  \sumtn    g_1(\X_t, \theta^\circ) \right\| \topr 0, \label{tilde_g_1-g_1}\\
    \sup_{|\theta-\theta^\circ|\leq  C n^{-1/2} } & n \left| {n^{-1}} \sumtn \tilde{e}_t(\theta) - {{f_u(\xi^\circ) \over2} } (\theta-\theta^\circ)^T J(\tau) (\theta-\theta^\circ) \right| \topr 0, \label{tilde_e-e}\\
    \sup_{|\theta-\theta^\circ|\leq C n^{-1/2} }
    & \left\| {n^{-1/2}} \sumtn \left\{   \tilde{g}_{1t}( \theta ) -  \tilde{g}_{1t}( \theta^\circ )
    - E\left[ \tilde{g}_{1t}(\theta) |\Ft \right]
    \right\}
    \right\| \topr 0\label{tilde_R_1n = o_P(1)}
\end{align}
 for any constant $C>0$.

We first verify (\ref{tilde_g_1-g_1}). Since $\{ \partial
q_t(\theta^\circ) /\partial\theta \}_{t\in\mathbb{Z}}$ is stationary
and ergodic, it follows from {\bf (N3)} that $n^{-1/2} \max_{1\leq t
\leq n} \|\partial q_t(\theta^\circ) /\partial\theta\| = o(1)$ a.s.
Thus, from {\bf (N4)}, we have
\begin{equation*}
\begin{aligned}
    &\left\| { n^{-1/2} } \sumtn \tilde{g}_{1t}(\theta^\circ) -{ n^{-1/2} }  \sumtn    g_1(\X_t, \theta^\circ) \right\| \\
    &\leq { n^{-1/2} }  \sumtn \left\| \frac{\partial \tilde{q}_t(\theta^\circ)}{\partial\theta} - \frac{\partial q_t(\theta^\circ)}{\partial\theta} \right\|
    + \left( \max_{1\leq t \leq n}  { n^{-1/2} }  \left\| \frac{\partial q_t(\theta^\circ)}{\partial\theta} \right\| \right)
    \sum_{t=1}^\infty \left| I(Y_t<\tilde{q}_t(\theta^\circ)) - I(Y_t< q_t(\theta^\circ)) \right| \\
    &\leq \op(1) + \op(1) \cdot \sum_{t=1}^\infty \left| I(Y_t<\tilde{q}_t(\theta^\circ)) - I(Y_t< q_t(\theta^\circ)) \right| .
\end{aligned}
\end{equation*}
On the other hand, by virtue of {\bf (C5)},  {\bf (N1)}
and the mean value theorem, we can have
\begin{align*}
    E\left[ \left| I(Y_t<\tilde{q}_t(\theta^\circ)) - I(Y_t< q_t(\theta^\circ)) \right| |\Ft \right]
    &= \left| F_u\left(\xi^\circ + \{ \tilde{q}_t(\theta^\circ) - q_t(\theta^\circ) \} / h_t({\vartheta^\circ}) \right) - F_u(\xi^\circ) \right| \\
    &\leq c_0^{-1} \|f_u\|_\infty \left| \tilde{q}_t(\theta^\circ) - q_t(\theta^\circ) \right|.
\end{align*}
Thus, by using {\bf(C6)} and Corollary~2.3 of \cite{HallHeyde1980}, we obtain (\ref{tilde_g_1-g_1}).

Next, we verify (\ref{tilde_e-e}). Owing to (\ref{e_t theta}), we
have
\begin{multline*}
    \tilde{e}_t(\theta) - e_t(\theta)
    = (\theta-\theta^\circ)^T
    \int_0^1
    \frac{\partial \tilde{q}_t(\theta(u))}{\partial\theta} \left\{ F_u\left(\xi^\circ + \frac{\tilde{q}_t(\theta(u))-q_t(\theta^\circ)}{h_t({\alpha^\circ})} \right) - \tau \right\} \\
    -  \frac{\partial q_t(\theta(u))}{\partial\theta} \left\{ F_u\left(\xi^\circ + \frac{q_t(\theta(u))-q_t(\theta^\circ)}{h_t({\alpha^\circ})} \right) - \tau \right\} \md u.
\end{multline*}
Similarly to the case of (\ref{tilde_g_1-g_1}), we have that for all
$n>\delta^{-2} C^2$,
\begin{align}\label{tilde_e-e sub}
\begin{aligned}
    &\sup_{\|\theta-\theta^\circ\|\leq  C n^{-1/2} } \left| \sumtn \tilde{e}_t(\theta) - \sumtn e_t(\theta) \right|\\
    &\leq { C n^{-1/2} } \sumtn \sup_{\|\theta-\theta^\circ \|\leq  C n^{-1/2} }
    \left\| \frac{\partial \tilde{q}_t(\theta)}{\partial\theta} - \frac{\partial q_t(\theta)}{\partial\theta} \right\| \\
    &\quad+ { C n^{-1/2} c_0^{-1} \|f_u\|_\infty}
    \sumtn \sup_{\|\theta-\theta^\circ \|\leq  C n^{-1/2} } \left\| \frac{\partial q_t(\theta)}{\partial\theta} \right\|
    \cdot
    \sup_{\|\theta-\theta^\circ \|\leq  C n^{-1/2} }  \left\|  \frac{\partial \tilde{q}_t(\theta)}{\partial\theta} - \frac{\partial q_t(\theta)}{\partial\theta} \right\|\\
    &\leq  \op(1) + {C c_0^{-1}  \|f_u\|_\infty} \left( \max_{1\leq t\leq n} { n^{-1/2} M_{1t} }  \right)
    \sumtn  \sup_{\|\theta-\theta^\circ \|\leq  C n^{-1/2} }  \left\|  \frac{\partial \tilde{q}_t(\theta)}{\partial\theta} - \frac{\partial q_t(\theta)}{\partial\theta} \right\| \\ &=\op(1).
\end{aligned}
\end{align}
This together with (\ref{Calculation of cond mean}) implies
(\ref{tilde_e-e}).

Finally, we deal with (\ref{tilde_R_1n = o_P(1)}). In this case, we
use arguments similar to those in Proposition~1 of \cite{bai1994}. Set
$z= n^{1/2} (\theta-\theta^\circ)$. For notational convenience, let
\begin{align*}
    \zeta_{nt}(z) &= \tilde{q}_t(\theta^\circ + n^{-1/2} z),&
    \zeta_{nt}^*(z) &= \frac{\partial \tilde{q}_t\left(\theta^\circ + n^{-1/2} z \right) }{\partial\theta}, &
    F_t(x) &= F_u\left(\xi^\circ + \frac{x-q_t(\theta^\circ)}{h_t({\alpha^\circ})} \right),
\end{align*}
and
\begin{align*}
    H_n(z)
    &= { n^{-1/2} } \sumtn \left[  \zeta_{nt}^*(z)
    \left\{ I\left( Y_t < \zeta_{nt}(z) \right) - F_t\left(\zeta_{nt}(z) \right)   \right\}
    - \zeta_{nt}^{*}(0) \left\{ I\left( Y_t < \zeta_{nt}(0) \right) - \tau \right\}
    \right].
\end{align*}
Then, (\ref{tilde_R_1n = o_P(1)}) is equivalent to $\sup_{\|z\|\leq C}
\|H_n(z)\|=\op(1)$, where  $C$ is any positive real number. Given any
 $\eta>0$, there exist a partition $\{ B_i : i=1, 2, \ldots,
m(\eta) \}$ of $\{z : \|z\|\leq C\}$ such that the diameter of each
$B_i$ is less than $\eta$. For each $i=1, 2, \ldots, m(\eta)$, pick
up $z_i\in B_i$. Note that for all $z\in B_i$ and  $n > \delta^{-2}
C^2$,
\begin{align*}
    I\left(Y_t < \zeta_{nt}(z_i) -  {\eta n^{-1/2} } \tilde{M}_{1t} \right)
    \leq I\left(Y_t < \zeta_{nt}(z) \right)
    \leq I\left(Y_t < \zeta_{nt}(z_i) +  {\eta n^{-1/2} } \tilde{M}_{1t} \right),
\end{align*}
where $\tilde{M}_{1t} \equiv \sup_{\theta\in N_\delta}
\|\partial\tilde{q}_t(\theta) / \partial\theta\|$. Therefore,  for all
large $n$,
\begin{align}\label{monotonicity for I(zeta_nt)}
\begin{aligned}
    &\sup_{z\in B_i}  \left\| \zeta_{nt}^*(z) I\left( Y_t < \zeta_{nt}(z) \right) - \zeta_{nt}^*(z_i) I\left( Y_t < \zeta_{nt}(z_i) \right)
    \right\| \\
    &\leq \sup_{z\in B_i} \left\| \zeta_{nt}^*(z) - \zeta_{nt}^*(z_i) \right\| + \left\| \zeta_{nt}^*(z_i) \right\| \cdot
    \sup_{z\in B_i}  \left| I\left( Y_t < \zeta_{nt}(z) \right) - I\left( Y_t < \zeta_{nt}(z_i) \right)
    \right| \\
    &\leq { \eta n^{-1/2} \tilde{M}_{2t} }  + \tilde{M}_{1t} \left\{ I\left( Y_t < \zeta_{nt}(z_i)  + {\eta n^{-1/2} } \tilde{M}_{1t}  \right)
    -  I\left( Y_t < \zeta_{nt}(z_i)  - {\eta n^{-1/2} } \tilde{M}_{1t}  \right)
    \right\},
\end{aligned}
\end{align}
where $\tilde{M}_{2t} \equiv \sup_{\theta\in N_\delta} \|\partial^2 \tilde{q}_t(\theta) / \partial\theta\partial\theta^T \|$.
Similarly, it follows from the mean value theorem that for all large $n$,
\begin{align}\label{monotonicity for F(zeta_nt)}
    \sup_{z\in B_i}  \left\| \zeta_{nt}^*(z) F_t\left( \zeta_{nt}(z) \right)
    - \zeta_{nt}^*(z_i) F_t\left( \zeta_{nt}(z_i) \right)
    \right\|
    &\leq {\eta n^{-1/2} \tilde{M}_{2t} }  + {\eta c_0^{-1} \|f_u\|_\infty n^{-1/2} \tilde{M}_{1t}^2 }.
\end{align}
From {\bf (N3)}(ii) and {\bf (N4)}(ii), we can have that
\begin{align*}
    {1\over n} \sumtn \tilde{M}_{1t}^2 \leq   {2\over n} \sumtn \sup_{\theta\in N_\delta} \left\| \frac{\partial q_t(\theta)}{\partial\theta} -  \frac{\partial\tilde{q}_t(\theta)}{\partial\theta} \right\|^2
    +  {2\over n} \sumtn M_{1t}^2 = O_p(1).
\end{align*}
It also follows from {\bf(N3)}(iii) and {\bf(N4)}(iii) that $n^{-1} \sumtn \tilde{M}_{2t}=\Op(1)$.

Define
\begin{align*}
    \mathcal{H}_n(z_i, \eta)
    &= { n^{-1/2} } \sumtn \tilde{M}_{1t} \left\{ I\left( Y_t < \zeta_{nt}(z_i)  + {\eta n^{-1/2} } \tilde{M}_{1t}  \right)
    - F_t\left( \zeta_{nt}(z_i)  + {\eta n^{-1/2} } \tilde{M}_{1t}  \right)
    \right\}.
\end{align*}
Then, it follows from (\ref{monotonicity for I(zeta_nt)}) and (\ref{monotonicity for F(zeta_nt)}) that for all large $n$,
\begin{align*}
    \sup_{z\in B_i} \|H_n(z)-H_n(z_i)\|
    &\leq
    \mathcal{H}_n(z_i, \eta) - \mathcal{H}_n(z_i, -\eta) \\
    &\quad + { n^{-1/2} } \sumtn \tilde{M}_{1t} \left\{
    F_t\left( \zeta_{nt}(z_i)  + {\eta n^{-1/2} } \tilde{M}_{1t}  \right)
    - F_t\left( \zeta_{nt}(z_i)  - {\eta n^{-1/2} } \tilde{M}_{1t}  \right)
    \right\} \\
    &\quad  + {2\eta n^{-1} } \sumtn \tilde{M}_{2t} + {\eta c_0^{-1} \|f_u\|_\infty  }{ n^{-1} } \sumtn \tilde{M}_{1t}^2  \\
    &\leq \left| \mathcal{H}_n(z_i, \eta) - \mathcal{H}_n(z_i, -\eta) \right| + \eta \Op(1),
\end{align*}
where the $\Op(1)$ does not depend on $z_i$.
Therefore,
\begin{align*}
    \sup_{\|z\|\leq C} \|H_n(z)\|
    &\leq \max_{1\leq i \leq m(\eta)} \sup_{z\in B_i} \|H_n(z)-H_n(z_i)\|  + \max_{1\leq i \leq m(\eta)} \|H_n(z_i)\| \\
    &\leq  \max_{1\leq i \leq m(\eta)} \left| \mathcal{H}_n(z_i, \eta) - \mathcal{H}_n(z_i, -\eta) \right|
    + \max_{1\leq i \leq m(\eta)} \|H_n(z_i)\| +  \eta \Op(1).
\end{align*}

Now, we only have to show that $ \mathcal{H}_n(z_i, \eta) -
\mathcal{H}_n(z_i, -\eta)$ and $H_n(z_i)$ converge to $0$ in
probability for each $\eta$ and $z_i$. Let $H_n(z_i)=\sumtn
\chi_{nt}(z_i)$.
First, it can be seen from {\bf(N4)} that $n^{-1/2} \max_{1\leq t \leq n}  \tilde{M}_{1t} = o(1) $ a.s.
Similarly, it follows that $n^{-1} \max_{1\leq t \leq n}  \tilde{M}_{2t}  =o(1)$ a.s.
Thus, we have
\begin{align}\label{E X_nt given}
\begin{aligned}
    \sumtn \left\| E\left[ \chi_{nt}(z_i) |\Ft \right] \right\|
    &= { n^{-1/2} } \sumtn \left\| \zeta_{nt}^{*}(0) \right\| \left| F_t\left( \zeta_{nt}(0) \right) - \tau \right| \\
    &\leq c_0^{-1} \|f_u\|_\infty  \left( \max_{1\leq t \leq n} n^{-1/2} \left\| \frac{\partial \tilde{q}_t(\theta^\circ)}{\partial\theta} \right\|  \right) \sumtn { \left| \tilde{q}_t(\theta^\circ) - q_t(\theta^\circ) \right|} \\
    &=\op(1) \Op(1).
\end{aligned}
\end{align}
Further, simple algebras show that
\begin{align*}
    &  n E\left[ \left\| \chi_{nt}(z_i)  \right\|^2 |\Ft \right] \\
    & = n E\left[ \left( \chi_{nt}(z_i) - E\left[ \chi_{nt}(z_i) |\Ft \right] \right)^T \chi_{nt}(z_i) \big|\Ft \right]
    + n \left\| E\left[ \chi_{nt}(z_i) |\Ft \right] \right\|^2   \\
    &\leq   E\left[ \left\| \zeta_{nt}^*(z_i) I\left( Y_t < \zeta_{nt}(z_i) \right)
    - \zeta_{nt}^{*}(0) I\left( Y_t < \zeta_{nt}(0) \right) \right\|^2
    | \Ft  \right]  +  \left\| \zeta_{nt}^{*}(0) \right\|^2 \left\{ F_t\left( \zeta_{nt}(0) \right) - \tau \right\}^2 \\
    &\leq 2   \left\| \zeta_{nt}^*(z_i) - \zeta_{nt}^*(0)  \right\|^2
    + 2  \left\| \zeta_{nt}^*(0) \right\|^2 \left| F_t\left( \zeta_{nt}(z_i) \right) - F_t\left( \zeta_{nt}(0) \right) \right|
    +  \left\| \zeta_{nt}^{*}(0) \right\|^2 \left\{ F_t\left( \zeta_{nt}(0) \right) - \tau \right\}^2 .
\end{align*}
Hence, it follows that for all large $n$,
\begin{align}\label{E X_nt^2 given}
\begin{aligned}
    \sumtn E\left[ \left\| \chi_{nt}(z_i)  \right\|^2 |\Ft \right]
    &\leq 2 C^2 \left( \max_{1\leq t \leq n} n^{-1} \tilde{M}_{2t} \right) n^{-1}\sumtn \tilde{M}_{2t} \\
    &\quad + 2 C c_0^{-1} \|f_u\|_\infty  \left( \max_{1\leq t \leq n} n^{-1/2} \tilde{M}_{1t} \right)
    n^{-1}\sumtn  \tilde{M}_{1t}^2 \\
    &\quad + c_0^{-2} \|f_u\|_\infty^2  \left( \max_{1\leq t \leq n} n^{-1} \tilde{M}_{1t}^2 \right) \sumtn \left| \tilde{q}_t(\theta^\circ) - q_t(\theta^\circ) \right|^2  \\
    &=\op(1).
\end{aligned}
\end{align}
By applying Lemma~9 of \cite{Genon-CatalotJacod1993} componentwise
together with (\ref{E X_nt given}) and (\ref{E X_nt^2 given}),
we have $H_n(z_i)=\op(1)$ for each $z_i$. Further, it can readily seen
that $ \mathcal{H}_n(z_i, \eta) - \mathcal{H}_n(z_i, -\eta)=\op(1)$
for each $\eta$ and $z_i$. Therefore, we get (\ref{tilde_R_1n =
o_P(1)}), which completes the proof. \hfill \fbox{}

\ \\ \indent \textbf{Proof of Lemma~\ref{identification lemma}}. For
the ARMA and AGARCH parameters $\varphi$ and $\vartheta$, we denote
 $\phi(z) = 1-\sum_{j=1}^P \phi_j z^j$,
 $\beta(z) =1-\sum_{j=1}^p \beta_j
z^j$, and $\gamma_l(z) = \sum_{i=1}^q \gamma_{li} z^i$ for $l=1, 2$.
Since we can write $q_t(\theta) = Y_t - \eps_t(\varphi) + \xi
h_t(\varphi,\vartheta)$, the condition can be reexpressed as
\begin{align}\label{identification equation}
    \eps_t(\varphi) - \eps_t - \xi h_t(\vpt) + \xi^\circ h_t = 0 \quad\text{a.s.}
\end{align}
for some $t\in\mathbb{Z}$ and $\theta\in\Theta$. Put
$\Delta_t(\varphi) = \eps_t(\varphi)-\eps_t$. Due to {\bf (A4)},
we can express 
\begin{align*}
    \Delta_t(\varphi) = a_0 + \sum_{j=1}^\infty a_j \eps_{t-j},
\end{align*}
where $a_0 = \{\psi(1) \phi^\circ(1)\}^{-1} {\phi(1) \phi_0^\circ} - \psi(1)^{-1} \phi_0$ and
$1+\sum_{j=1}^\infty a_j z^j = \{\psi(z) \phi^\circ(z)\}^{-1} {\phi(z) \psi^\circ(z)}$ for $|z|\leq1$.
Further,
\begin{align*}
    h_t^2(\vpt) &= { c_0 +
    \sum_{j=1}^\infty c_{1j} {( \eps_{t-j} + \Delta_{t-j}(\varphi) )^+}^2
    + \sum_{j=1}^\infty c_{2j} {( \eps_{t-j} + \Delta_{t-j}(\varphi) )^-}^2 }, \\
    h_t^2 &= { c_0^\circ +
    \sum_{j=1}^\infty c_{1j}^\circ ( \eps_{t-j}^+ )^2
    + \sum_{j=1}^\infty c_{2j}^\circ ( \eps_{t-j}^- )^2,
    }
\end{align*}
where $c_0 := c_0(\vartheta)= 1/\beta(1)$, $c_0^\circ := c_0(\vartheta^\circ)$,
$c_{lj} := c_{lj}(\vartheta)$,  $c_{lj}^\circ := c_{lj}(\vartheta^\circ)$, and
$
    \sum_{j=1}^\infty c_{lj}(\vartheta) z^j = {\beta(z)}^{-1} {\gamma_l(z)}
$ for $|z|\leq1$ and $l=1, 2$. Then, by dividing
equation~(\ref{identification equation}) by $h_{t-1}$ and expressing
it as a function of $u_{t-1}$, we rewrite (\ref{identification
equation}) as
\begin{align}\label{identification eqn2}
\begin{aligned}
    &f_1(u_{t-1}, A_{t,2}, B_{t,2}, C_{t,2}, D_{t,2} ) \\
    &:= a_1 u_{t-1} + A_{t,2} - \xi \left\{ c_{11} {(u_{t-1}+B_{t,2})^+}^2 + c_{21} {(u_{t-1}+B_{t,2})^-}^2
    + C_{t,2}   \right\}^{1/2} \\
    &\quad +\xi^\circ \left\{ c_{11}^\circ (u_{t-1}^+)^2 + c_{21}^\circ (u_{t-1}^-)^2
    + D_{t,2}   \right\}^{1/2} \\
    &=0 \quad\text{a.s.,}
\end{aligned}
\end{align}
where for $k\geq2$,
\begin{align*}
    A_{t,k} &= \left(  a_0 + \sum_{j=k}^\infty a_j \eps_{t-j} \right) / h_{t-k+1},  \qquad\qquad\qquad
    B_{t,k} =  \Delta_{t-k+1}(\varphi) / h_{t-k+1},  \\
    C_{t,k} &= \left( c_0 +   \sum_{j=k}^\infty c_{1j} ( \eps_{t-j}^+(\varphi) )^2
    + \sum_{j=k}^\infty c_{2j} ( \eps_{t-j}^-(\varphi) )^2
    \right) /h_{t-k+1}^2, \\
    D_{t,k} &= \left( c_0^\circ +    \sum_{j=k}^\infty c_{1j}^\circ ( \eps_{t-j}^+ )^2
    + \sum_{j=k}^\infty c_{2j}^\circ ( \eps_{t-j}^- )^2 \right) /h_{t-k+1}^2.
\end{align*}
Note that $A_{t,k}, B_{t,k}, C_{t,k}, D_{t,k}$ are
$\mathcal{F}_{t-k}$-measurable. Since $f_1(\cdot)$ is continuous, it
follows from (\ref{identification eqn2}) that $f_1(y)=0$ for any $y$
in the support of the random vector $(u_{t-1}, A_{t,2}, B_{t,2},
C_{t,2}, D_{t,2} )$. Here, the independence of $u_{t-1}$ and
$\mathcal{F}_{t-2}$ and {\bf (A5)} indicates that the above support
is a Cartesian product of $\mathbb{R}$ and the support of $(A_{t,2},
B_{t,2}, C_{t,2}, D_{t,2} )$. Subsequently, it must hold with
probability 1 that
\begin{align}\label{identification eqn3}
    f_1(x, A_{t,2}, B_{t,2}, C_{t,2}, D_{t,2} )=0 \quad\text{for all}~~ x\in\mathbb{R}.
\end{align}
Now, we consider the following identity: $\forall x\in\mathbb{R}$,
\begin{align}\label{the identity}
    f(x) := ax+A-\xi \left\{ c_1 {(x+B)^+}^2 + c_2 {(x+B)^-}^2 + C \right\}^{1/2}
    + \xi^\circ \left\{ c_1^\circ (x^+)^2 + c_2^\circ (x^-)^2 + D \right\}^{1/2} =0,
\end{align}
where coefficients are real numbers. Then, by taking the limit $x\to
\pm \infty$ of $x^{-1} f(x)$, we can check that the following holds:
\begin{align}\label{1st identity condition}
    a-\xi {c_1}^{1/2} + \xi^\circ {c_1^\circ}^{1/2} &=0  & \text{and} &&
    a+\xi {c_2}^{1/2} - \xi^\circ {c_2^\circ}^{1/2} &=0.
\end{align}

First, we consider the case that $\xi^\circ\ne0$. Let $m\geq1$ be
the smallest integer such that $c_{1m}^\circ+c_{2m}^\circ>0$. Assume
that $m>1$. Then, it holds that $c_{11}^\circ=c_{21}^\circ=0$. Note
that due to (\ref{identification eqn3}) and (\ref{1st identity
condition}), $a_1=\xi c_{11}^{1/2}= \xi c_{21}^{1/2}=0$ since
$c_{1j}$, $c_{2j}$, $j\geq1$, are nonnegative. Hence, it can be seen
by simple algebras that $a_k=\xi c_{1k}^{1/2}= \xi c_{2k}^{1/2}=0$
for $1\leq k < m$ and (\ref{identification eqn2}) is reduced to the
following:
\begin{align}\label{identification eqn4}
\begin{aligned}
    &f_m(u_{t-m}, A_{t,m+1}, B_{t,m+1}, C_{t,m+1}, D_{t,m+1}) \\
    &:= a_m u_{t-m} + A_{t,m+1} - \xi \left\{ c_{1m} {(u_{t-m}+B_{t,m+1})^+}^2 + c_{2m} {(u_{t-m}+B_{t,m+1})^-}^2
    + C_{t,m+1}   \right\}^{1/2} \\
    &\quad +\xi^\circ \left\{ c_{1m}^\circ (u_{t-m}^+)^2 + c_{2m}^\circ (u_{t-m}^-)^2
    + D_{t,m+1}   \right\}^{1/2} \\
    &=0 \quad\text{a.s.}
\end{aligned}
\end{align}
If $\xi=0$, by using (\ref{1st identity condition}) again, we get
$a_m=-\xi^\circ {c_{1m}^\circ}^{1/2}=\xi^\circ
{c_{2m}^\circ}^{1/2}$, which leads to a contradiction. Thus, it must
hold that $\xi\ne0$. Here, considering the identity in ~(\ref{the
identity}), assume that $f(x)=0$, $\forall x\in\mathbb{R}$, where
$c_1, c_2, c_1^\circ, c_2^\circ\geq 0$, $c_1^\circ+c_2^\circ>0$,
$\xi^\circ\ne0$, $\xi\ne0$, and $C, D>0$. Then, since $\md^2
f(x)/\md x^2=0$ for $x\ne0, -B$, $\lim_{x\downarrow \max\{0,-B\}}
\md^3 f(x) / \md x^3=0$, and $\lim_{x\uparrow \min\{0,-B\}} \md^3
f(x) / \md x^3=0$, one can check that $c_1+c_2>0$ and $B=0$,
which together with (\ref{identification eqn4}) yields $B_{t,m+1}=0$
a.s. Hence, $\eps_{t-m}(\varphi)=\eps_{t-m}$ a.s. and subsequently,
due to {\bf (A3)}(ii), we obtain $\varphi=\varphi^\circ$ and $a_j=0,
\forall j\geq0$. Since it can be further entailed that $A_{t,k}=0$
a.s. for all $k\geq2$ and $\xi {c_{lm}}^{1/2}= \xi^\circ
{c_{lm}^\circ}^{1/2}$ for $l=1,2$, (\ref{identification eqn4}) is
reduced to $\xi C_{t,m+1}^{1/2} = \xi^\circ D_{t,m+1}^{1/2}$ a.s.
Then, repeating the above steps, we are led to get $\xi
{c_{lj}}^{1/2}= \xi^\circ {c_{lj}^\circ}^{1/2}$ for all $j\geq1$ and
$l=1,2$. Since this directly implies $\xi {c_0}^{1/2} = \xi^\circ
{c_0^\circ}^{1/2}$, we have that  for $l=1,2$,
\begin{align*}
    \frac{\xi^2 \gamma_l(z)}{\beta(z)} &= \frac{ {\xi^\circ} ^2 \gamma_l^\circ(z)}{\beta^\circ(z)},   \quad |z|\leq1.
\end{align*}
Owing to this, by using {\bf (A3)}(i) and standard arguments (for
instance, those in \citealt[p.~2481]{StraumannMikosch2006}), it can be
verified that $\beta(\cdot) \equiv \beta^\circ(\cdot)$. This entails
$\xi=\xi^\circ$ and thus,
$\gamma_l(\cdot)\equiv\gamma_l^\circ(\cdot)$, $l=1,2$.

We now consider the case that $\xi^\circ=0$. If $\xi=0$, it follows
from (\ref{identification equation}) that $\varphi=\varphi^\circ$.
Suppose that $\xi\ne0$. Then, due to (\ref{1st identity condition})
and (\ref{identification eqn2}), we get $a_1=c_{11}=c_{21}=0$.
Similarly to the previous case as above, we can see that the
repeated arguments yield  $a_j=c_{1j}=c_{2j}=0$ for all $j\geq1$ and
therefore, $\phi(\cdot)\equiv \phi^\circ(\cdot)$, $\psi(\cdot)\equiv
\psi^\circ(\cdot)$, and $\gamma_1(\cdot) \equiv \gamma_2(\cdot)
\equiv 0$. Then, combining all these and (\ref{identification
equation}), we can obtain $a_0 - \xi \beta(1)^{-1/2}=0$, which
completes the proof. \hfill \fbox{}

\ \\ \indent \textbf{Proof of Lemma~\ref{positive definiteness}}. It
suffices to show that $\lambda^T ( \partial q_1(\theta^\circ)
/\partial\theta )=0$ a.s. for some
$\lambda\in\mathbb{R}^{P+Q+2+p+2q}$ implies $\lambda=0$. Suppose
that $\lambda^T ( \partial q_1(\theta^\circ) /\partial\theta )=0$
almost surely. Then, by the stationarity, $\lambda^T ( \partial
q_t(\theta^\circ) /\partial\theta )=0$ a.s. for all
$t\in\mathbb{Z}$. In view of (\ref{derivative of q_t}), we can
express $2 h_t \lambda^T ( \partial q_t(\theta^\circ)
/\partial\theta )=0$ as
\begin{align}\label{p.d. identity1}
    2\lambda_1 h_t^2 + \lambda_2^T \left( - 2 h_t \frac{\partial \eps_t(\varphi^\circ)}{\partial\varphi}
    + \xi^\circ \frac{\partial h_t^2(\varphi^\circ, \vartheta^\circ)}{\partial\varphi} \right) + \xi^\circ \lambda_3^T \frac{\partial h_t^2(\varphi^\circ, \vartheta^\circ)}{\partial\vartheta} =0,
\end{align}
where $\lambda^T=(\lambda_1, \lambda_2^T, \lambda_3^T)$,
$\lambda_1\in\mathbb{R}$, $\lambda_2=(\lambda_{2,1}, \lambda_{2,2},
\ldots, \lambda_{2,P+Q+1})^T\in\mathbb{R}^{P+Q+1}$ and
$\lambda_3=(\lambda_{3,1}, \lambda_{3,2}, \ldots,
\lambda_{3,2q+p})^T\in\mathbb{R}^{2q+p}$. To see that $\lambda=0$,
we use the same techniques to prove Lemma~\ref{identification
lemma}.  We express (\ref{p.d. identity1}) as a function of
$u_{t-1}$ as follows:
\begin{align}\label{pd ident-1}
    \frac{2 \lambda_1 h_t^2}{h_{t-1}^2} &= 2\lambda_1 \left\{ \gamma_{11}^\circ (u_{t-1}^+)^2 + \gamma_{21}^\circ (u_{t-1}^-)^2 + D_{t,2} \right\},
\end{align}
where $D_{t,2}$ is defined in the proof of Lemma~\ref{identification
lemma}. Similarly, due to (\ref{derivative of q_t}), we can express
\begin{align}\label{pd ident-2}
\begin{aligned}
    \left\{ {1\over h_{t-1}} (-2) \lambda_2^T  \frac{\partial \eps_t(\varphi^\circ)}{\partial\varphi} \right\}
    {h_t\over h_{t-1}}
    &= \left\{ 2(\lambda_{2,2}+\lambda_{2,P+2})u_{t-1} + E_{t,2} \right\}
    \left\{ \gamma_{11}^\circ (u_{t-1}^+)^2 + \gamma_{21}^\circ (u_{t-1}^-)^2 + D_{t,2} \right\}^{1/2}, \\
    {1\over h_{t-1}^2} \xi^\circ \lambda_2^T \frac{\partial h_t^2(\varphi^\circ, \vartheta^\circ)}{\partial\varphi}
    &= 2 \xi^\circ \left( \gamma_{11}^\circ u_{t-1}^+ - \gamma_{21}^\circ u_{t-1}^- \right) {1\over h_{t-1}} \lambda_2^T
    \frac{\partial \eps_{t-1}(\varphi^\circ)}{\partial\varphi}  + F_{t,2}, \\
    {1\over h_{t-1}^2} \xi^\circ \lambda_3^T \frac{\partial h_t^2(\varphi^\circ, \vartheta^\circ)}{\partial\vartheta}
    &= \xi^\circ \left\{ \lambda_{3,1} (u_{t-1}^+)^2 + \lambda_{3,q+1} (u_{t-1}^-)^2 \right\} +
    G_{t,2},
\end{aligned}
\end{align}
where $E_{t,2}, F_{t,2}, G_{t,2}$ are obviously defined and
$\mathcal{F}_{t-2}$-measurable. Then, due to (\ref{p.d. identity1}),
one can see that the sum of terms in the right-hand side of (\ref{pd
ident-1}) and (\ref{pd ident-2}) equals to $0$ a.s. Then, by using
the same arguments deducing (\ref{identification eqn3}), it can be
seen that  with probability 1, for all $x>0$,
\begin{align*}
    f(x) :=& (2\lambda_1\gamma_{11}^\circ + \xi^\circ\lambda_{3,1}) x^2
    + \left\{ {2 \xi^\circ \gamma_{11}^\circ  \over h_{t-1}} \lambda_2^T
    \frac{\partial \eps_{t-1}(\varphi^\circ)}{\partial\varphi} \right\} x
    + \left\{ 2\lambda_1 D_{t,2} + F_{t,2}+G_{t,2} \right\} \\
    &+ \left\{ 2(\lambda_{2,2}+\lambda_{2,P+2}) x + E_{t,2} \right\}
    \left\{ \gamma_{11}^\circ x^2  + D_{t,2} \right\}^{1/2} \\
    =&0.
\end{align*}
Here, note that $D_{t,2}>0$ a.s. and $\xi^\circ \ne 0$,
$\gamma_{11}^\circ >0$. Then, since $\lim_{x\to0} \md f(x)/\md x=0$
and $\lim_{x\to0} \md^3 f(x)/\md x^3=0$, we can have
\begin{align}\label{pd iden lambda_2}
    \lambda_2^T \frac{\partial \eps_{t-1}(\varphi^\circ)}{\partial\varphi} =0 \quad\text{a.s.}
\end{align}
As in \citet[p.~631]{FrancqZakoian2004}, we can also see that under
the minimality assumption {\bf (A3)}(ii) on the ARMA representation,
(\ref{pd iden lambda_2}) implies $\lambda_2=0$.

Now, (\ref{p.d. identity1}) is reduced to
\begin{align}\label{pd ident-3}
    2\lambda_1 h_t^2 + \xi^\circ \lambda_3^T \frac{\partial h_t^2(\varphi^\circ, \vartheta^\circ)}{\partial\vartheta} =0 \quad\text{a.s.}
\end{align}
Note that Lemma~\ref{identification lemma} entails
$h_t^2=(const.)\times h_t^2(\varphi^\circ, \vartheta)$ a.s. implies
$const.=1$ and $\vartheta=\vartheta^\circ$, which implies that the
representation (\ref{AGARCH model eqn}) is minimal particularly
under {\bf (A3)}(i). From (\ref{def of h_t(vpt)}), we have
\begin{align}\label{beta(B) h_t^2 derivative}
    \beta^\circ(B) \frac{\partial h_t^2(\varphi^\circ, \vartheta^\circ)}{\partial\vartheta}
    = \left( (\eps_{t-1}^+)^2, \ldots, (\eps_{t-q}^+)^2, (\eps_{t-1}^-)^2, \ldots, (\eps_{t-q}^-)^2, h_{t-1}^2, \ldots, h_{t-p}^2
    \right)^T.
\end{align}
Using the fact that $a(x^+)^2 + b(x^-)^2+c=0$, $\forall
x\in\mathbb{R}$ implies $a=b=0$, we can see that any constants and
random variables in the right-hand side of (\ref{beta(B) h_t^2
derivative}) are linearly independent due to {\bf(A3)}(i) (see the
arguments in \citealt[p.~621]{FrancqZakoian2004}). Since (\ref{pd
ident-3}) implies
\begin{align*}
    &2\lambda_1 \left( 1+ \sum_{i=1}^q \gamma_{1i}^\circ (\eps_{t-i}^+)^2 + \sum_{i=1}^q \gamma_{2i}^\circ (\eps_{t-i}^-)^2
    \right) + \xi^\circ \lambda_3^T \beta^\circ(B) \frac{\partial h_t^2(\varphi^\circ, \vartheta^\circ)}{\partial\vartheta} \\
    &= \left\{ \left( 2\lambda_1, 2\lambda_1 \gamma_{11}^\circ,  \ldots,
    2\lambda_1 \gamma_{1q}^\circ, 2\lambda_1 \gamma_{21}^\circ, \ldots, 2\lambda_1 \gamma_{2q}^\circ, 0, \ldots, 0 \right)
    + \left( 0, \xi^\circ \lambda_3^T \right)
    \right\}  \\
    &\quad\cdot \left( 1, (\eps_{t-1}^+)^2, \ldots, (\eps_{t-q}^+)^2, (\eps_{t-1}^-)^2, \ldots, (\eps_{t-q}^-)^2, h_{t-1}^2, \ldots, h_{t-p}^2
    \right)^T \\
    &=0 \quad\text{a.s.},
\end{align*}
we obtain $\lambda_1=0$ and so $\lambda_3=0$, which completes the
proof.
\hfill \fbox{} \\

In what follows, $\rho\in(0,1)$, $K>0$ and $V$ denote generic constants
and a generic random variable, respectively, which may vary from line to line. Further,
$\{S_t\}$ denotes a generic stationary ergodic process such that
$S_t\in\mathcal{F}_{t-1}$ and $E[S_t^2]<\infty$.

\ \\ \indent \textbf{Proof of Theorem~\ref{consistency for AGARCH}}.
Note that {\bf (A5)} is sufficient for {\bf (C1)} and that {\bf(C5)}
trivially holds with $h_1(\alpha^\circ)\geq1$ a.s.
We now verify {\bf(C3)} and {\bf(C6)}.
Recall that $E|\eps_t|<\infty$ is equivalent to $E|Y_t|<\infty$
under {\bf (A2)}. For
any analytic function $f(z)$ on $\{ z\in\mathbb{C}: |z|\leq1 \}$, we
denote by $a_k(f)$ the coefficient of $z^k$ in its Taylor's series
expansion. Due to {\bf (A4)}, we can express
\begin{align}
    \eps_t(\varphi) &= - \psi(1)^{-1} \phi_0 +   \sum_{k=0}^\infty a_k( \phi / \psi ) Y_{t-k},
    \label{eps filter representation} \\
    h_t(\vpt) &= \left\{ \beta(1)^{-1} +  \sum_{k=1}^{\infty} a_k(\gamma_1/\beta) ( \eps_{t-k}^+(\varphi) )^2
    + \sum_{k=1}^{\infty} a_k(\gamma_2/\beta) ( \eps_{t-k}^-(\varphi) )^2
    \right\}^{1/2}. \label{h_t filter representation}
\end{align}
Thus, it can be seen that $q_t(\theta) = Y_t - \eps_t(\varphi) + \xi
h_t(\vpt)$ is of the form in (\ref{def of q_t(theta)}). For any
polynomial $\psi(\cdot)$ of degree $Q$, we define $\rho(\psi) = \max
\{ |z_i|^{-1} : \psi(z_i)=0, z_i\in\mathbb{C}, 1\leq i\leq Q \}$.
Note that $\sum_{j=1}^p \beta_j<1$ implies $\rho(\beta)<1$ (see Lemma~2.1 of \citealt{Berkesetal2003}).
Due to {\bf (A4)} and the compactness of
$\Theta$, we have that $\sup_{\theta\in\Theta} \rho(\psi) <1$ and
$\sup_{\theta\in\Theta} \rho(\beta) <1$, from which it can be shown
that $\sup_{\theta\in\Theta} |a_k(1/\psi) | \leq K \rho^k$ and
$\sup_{\theta\in\Theta} |a_k(1/\beta) | \leq K \rho^k$ for all
$k\geq0$ (see, e.g., Theorem~3.1.1 of \citealt{BrockwellDavis1991}).
Further, one can see that $a_k(\phi/\psi)$, $a_k(\gamma_1/\beta)$
and $a_k(\gamma_2/\beta)$  decay exponentially fast uniformly on
$\Theta$. By using these, the fact that $E|Y_t| <\infty$, (\ref{eps
filter representation}), and (\ref{h_t filter representation}), we
have
\begin{align}\label{bdd for eps_t, h_t}
\begin{aligned}
    \sup_{\theta\in\Theta} |\eps_t(\varphi)| &\leq K+ \sum_{k=0}^\infty K\rho^k |Y_{t-k}| = S_{t+1}^2, \\
    \sup_{\theta\in\Theta} |h_t(\vpt)| &\leq \left\{ K+ \sum_{k=1}^\infty K\rho^k  S_{t-k+1}^4 \right\}^{1/2} \leq S_t^2,
\end{aligned}
\end{align}
which ensures {\bf (C3)}(ii).

Note that the recursion for $\tilde{\eps}_t(\varphi)$ in
Section~\ref{sec3} can be expressed as
\begin{align}\label{tilde eps backshift}
    \psi(B) \tilde\eps_t(\varphi) &= \phi(B) ( Y_t^* - \phi(1)^{-1} \phi_0 ),
\end{align}
where $B$ denotes the backshift operator, $Y_t^*=Y_t$ for $t\geq1$,
and $Y_t^*=\phi(1)^{-1} \phi_0$ for $t\leq0$. Then, owing to
(\ref{tilde eps backshift}), we have that for $t\geq1$,
\begin{align}\label{tilde eps filter representation}
    \tilde\eps_t(\varphi) &= \sum_{k=0}^{t-1} a_k(\phi/\psi) ( Y_{t-k} - \phi(1)^{-1}\phi_0).
\end{align}
Similarly, with the initial values of
$\tilde{h}_t^2(\vpt)=\beta(1)^{-1}$ for $t\leq0$, we can express
\begin{align}\label{tilde h_t^2 filter representation}
    \tilde{h}_t^2(\vpt) &= \beta(1)^{-1} + \sum_{k=1}^{t-1} a_k(\gamma_1/\beta) ( \tilde\eps_{t-k}^+(\varphi) )^2
    + \sum_{k=1}^{t-1} a_k(\gamma_2/\beta) ( \tilde\eps_{t-k}^-(\varphi) )^2.
\end{align}

It is easy to check that
\begin{align}
    \sup_{\theta\in\Theta} |\tilde\eps_t(\varphi)| &\leq \sum_{k=0}^\infty K\rho^k (|Y_{t-k}|+1) = S_{t+1}^2, \label{bddness of tilde eps_t} \\
    \sup_{\theta\in\Theta} \left| \eps_t(\varphi) -\tilde\eps_t(\varphi) \right|
    &\leq \rho^t \sum_{j=0}^\infty K \rho^j (|Y_{-j}| + 1)=V\rho^t.\label{closeness of eps_t and tilde}
\end{align}
Further, since $|(x^\pm)^2 - (y^\pm)^2|\leq |x-y|\{|x|+|y|\}$, it
can be easily seen that $
    \sup_{\theta\in\Theta} | (\eps_t^\pm(\varphi))^2 - ( \tilde\eps_t^\pm(\varphi) )^2 |
    \leq V \rho^t S_{t+1}^2 \leq V \rho^{t}$.
Since $E|S_t|<\infty$, we have $E \log^+|S_t|<\infty$ and thus, due
to Lemma~2.2 of \cite{Berkesetal2003}, $\sum_{j=0}^\infty \rho^j
S_{-j+1}^4$ converges with probability 1. Further, since $\min\{
h_t(\vpt), \tilde{h}_t(\vpt)\}\geq 1$ for all $\theta\in\Theta$, it
follows that
\begin{align}\label{closeness of h_t^2 and tilde}
    \sup_{\theta\in\Theta} \left| h_t(\vpt) - \tilde{h}_t(\vpt) \right|
    &\leq
    2^{-1} \sup_{\theta\in\Theta} \left| h_t^2(\vpt) - \tilde{h}_t^2(\vpt) \right| \nonumber \\
    &\leq \sum_{k=1}^{t-1} K\rho^k V \rho^{t-k} + \sum_{k=t}^\infty K\rho^k S_{t-k+1}^4 \leq V\rho^t.
\end{align}
This together with (\ref{closeness of eps_t and tilde}) implies {\bf
(C6)}, and henceforth, an application of Lemma~\ref{identification
lemma}(i) and Theorem~\ref{strong consistency} validates
Theorem~\ref{consistency for AGARCH}(i).

Next, we deal with the case when $\xi^\circ=0$. Since {\bf (C3)} and
{\bf (C6)} are satisfied,
$\tilde{G}_n(\theta)-\tilde{G}_n(\theta^\circ)$ uniformly converges
a.s. to $\Gamma(\theta)$, which is the one defined in the proof of
Theorem~\ref{strong consistency}.
Note that $\Gamma(\theta)\geq 0$ and $\Gamma(\theta)=0$ if and only if
$q_t(\theta)=Y_t - \eps_t(\varphi^\circ)$ in this case.
Then, it follows from Lemma~\ref{identification lemma}(ii) that
$\Gamma(\theta)=0$ implies $\phi_j=\phi_j^\circ, 1\leq j \leq P$ and $\psi_i=\psi_i^\circ, 1\leq i \leq Q$.
Due to the compactness of
$\Theta$, for each generic point $w$ of the underlying probability
space, there exists a subsequence $ \hat\theta_{n_k}
:=\hat\theta_{n_k}(w) $ tending to a limit
$\theta^\infty:=\theta^\infty (w)$. From the uniform convergence and
the continuity of $\Gamma(\theta)$, we have that $
\tilde{G}_{n_k}(\hat\theta_{n_k}) - \tilde{G}_{n_k}(\theta^\circ)
\to \Gamma(\theta^\infty)$ as $k\to\infty $. Since
$\tilde{G}_n(\hat\theta_n) \leq \tilde{G}_n(\theta^\circ)$ and
$\Gamma(\theta)\geq0$, we have $ \Gamma(\theta^\infty)=0$.
It follows from the above argument that
$\phi_j^\infty=\phi_j^\circ$, $1\leq j\leq P$ and
$\psi_i^\infty=\psi_i^\circ$, $1\leq i\leq Q$.
We have proved that any convergent subsequence of
$(\hat\phi_{1n}, \ldots, \hat\phi_{Pn}, \hat\psi_{1n}, \ldots, \hat\psi_{Qn}  )$
tends to the corresponding true parameter vector,
which validates
Theorem~\ref{consistency for AGARCH}(ii). \hfill \fbox{}

\ \\ \indent
\textbf{Proof of Theorem~\ref{Asymptotics for AGARCH}}.
In view of Theorem~\ref{consistency for AGARCH}, it remains to verify that assumptions {\bf (N3)}--{\bf(N5)} hold.
Recall that {\bf (A1')} implies $EY_t^2<\infty$.

Due to {\bf(N2)}, we can choose a neighborhood $N_\delta \subset
\Theta$ where $\gamma_{1i}$'s, $\gamma_{2i}$'s and $\beta_j$'s are
uniformly bounded away from $0$. From (\ref{def of eps_t(varphi)})
and (\ref{h_t filter representation}), the first derivatives of
$q_t(\theta)$ are given as follows:
\begin{align}\label{derivative of q_t}
    {\partial q_t(\theta)\over\partial\xi} &= h_t(\vpt), &
    {\partial q_t(\theta)\over\partial\varphi}  &= - {\partial \eps_t(\varphi)\over\partial\varphi} + \xi {\partial h_t(\vpt)\over\partial\varphi} , &
    {\partial q_t(\theta)\over\partial\vartheta} &= \xi {\partial h_t(\vpt)\over\partial\vartheta},
\end{align}
where
\begin{align*}
    {\partial \eps_t(\varphi)\over\partial\phi_0} &= - \psi(1)^{-1}, \qquad\qquad\qquad\qquad
    {\partial \eps_t(\varphi)\over\partial\phi_j} =  - \sum_{k=0}^\infty a_k(1/\psi) Y_{t-j-k}, \quad 1\leq j\leq P, \\
    {\partial \eps_t(\varphi)\over\partial\psi_i} &=   -\sum_{k=0}^\infty a_k(1/\psi) \eps_{t-i-k}(\varphi), \quad 1\leq i \leq Q, \\
    {\partial h_t^2(\vpt)\over\partial\varphi} &= 2 \sum_{k=1}^\infty  \left\{
    a_k(\gamma_1/\beta) \eps_{t-k}^+(\varphi) - a_k(\gamma_2/\beta) \eps_{t-k}^-(\varphi) \right\}
     {\partial \eps_{t-k}(\varphi)\over\partial\varphi}, \\
    {\partial h_t^2(\vpt)\over\partial\gamma_{1i}} &= \sum_{k=1}^\infty {\partial a_k(\gamma_1/\beta)\over\partial\gamma_{1i}}
    ( \eps_{t-k}^+(\varphi) )^2, \quad
    {\partial h_t^2(\vpt)\over\partial\gamma_{2i}} = \sum_{k=1}^\infty {\partial a_k(\gamma_2/\beta)\over\partial\gamma_{2i}}
    ( \eps_{t-k}^-(\varphi) )^2, \quad 1\leq i\leq q, \\
    {\partial h_t^2(\vpt)\over\partial\beta_j} &= -\beta(1)^{-2} +
    \sum_{k=1}^{\infty} {\partial a_k(\gamma_1/\beta)\over\partial\beta_j} ( \eps_{t-k}^+(\varphi) )^2
    + \sum_{k=1}^{\infty} {\partial a_k(\gamma_2/\beta)\over\partial\beta_j} ( \eps_{t-k}^-(\varphi) )^2, \quad 1\leq j\leq p.
\end{align*}
It can be seen that the above derivatives are all continuously differentiable
in $\theta\in N_\delta$ except for $\partial
h_t^2(\vpt)/\partial\varphi$.
In particular, $\partial^2
h_t^2(\vpt)/\partial\varphi\partial\varphi^T$ is discontinuous.
However, one can see that $\partial h_t^2(\vpt)/\partial\varphi$ is
Lipschitz continuous in $N_\delta$ and thus, {\bf(N3)}(i) is
satisfied.

Since $EY_t^2<\infty$, (\ref{bdd for eps_t, h_t}) becomes
$\sup_{\theta\in \Theta}|\eps_t(\varphi)| \leq S_{t+1}$ and
$\sup_{\theta\in \Theta} |h_t(\vpt)| \leq  S_t$. Thus, we have
$\sup_{\theta\in N_\delta} \|\partial\eps_t(\varphi)/\partial\varphi
\| \leq S_t$. Note that $a_k(\gamma_l/ \beta)\geq0$ for $k\geq1$ and
$l=1, 2$. Then, using $\partial h_t(\vpt) / \partial\theta = (2
h_t(\vpt))^{-1} \partial h_t^2(\vpt) / \partial\theta $, we get
\begin{align*}
    \left\| {\partial h_t(\vpt)\over\partial\varphi} \right\|
    &\leq  \sum_{k=1}^\infty \frac{
    \left\{  a_k(\gamma_1/\beta) \eps_{t-k}^+(\varphi) + a_k(\gamma_2/\beta) \eps_{t-k}^-(\varphi) \right\}
    \left\| {\partial \eps_{t-k}(\varphi)\over\partial\varphi} \right\| }
    {  \left\{  a_k(\gamma_1/\beta) (\eps_{t-k}^+(\varphi))^2 + a_k(\gamma_2/\beta) (\eps_{t-k}^-(\varphi))^2 \right\}^{1/2}
    } \\
    &\leq  \sum_{k=1}^\infty \left\{ a_k^{1/2}(\gamma_1/\beta) + a_k^{1/2}(\gamma_2/\beta) \right\} \left\| {\partial\eps_{t-k}(\varphi)\over\partial\varphi} \right\|.
\end{align*}
This in turn implies $ \sup_{\theta\in N_\delta} \|
\partial h_t(\vpt) /\partial\varphi \| \leq S_t$. Further, by virtue
of Lemma~3.2 of \cite{Berkesetal2003}, similarly, we can have
$\sup_{\theta\in N_\delta} \|
\partial h_t(\vpt) /\partial\vartheta \| \leq S_t$. Hence, {\bf (N3)}(ii) is satisfied.
On the other hand, simple algebras show that $\sup_{\theta\in
N_\delta}\|
\partial^2 \eps_t(\varphi) /\partial\varphi\partial\varphi^T\|\leq
S_t$. Then, using this and Lemma~3.3 of \cite{Berkesetal2003}, it can
be readily checked that $\sup_{\theta\in N_\delta}\| \partial^2
h_t^2(\vpt) /\partial\theta\partial\theta^T\|\leq S_t^2$. Hence, by
using {\bf(N3)}(ii) and the equality
\begin{align*}
    \frac{\partial^2 h_t(\vpt)}{\partial\theta\partial\theta^T}={1\over 2h_t(\vpt)}
    \frac{\partial^2 h_t^2(\vpt)}{\partial\theta\partial\theta^T} - {1\over h_t(\vpt)} \frac{\partial h_t(\vpt)}{\partial\theta}
    \frac{\partial h_t(\vpt)}{\partial\theta^T},
\end{align*}
we can see that {\bf (N3)}(iii) holds.

Meanwhile, owing to (\ref{tilde eps filter representation}),
(\ref{tilde h_t^2 filter representation}) and (\ref{derivative of
q_t}), we can derive
\begin{align*}
    \sup_{\theta\in N_\delta} \left\| \frac{\partial\eps_t(\varphi)}{\partial\varphi}
    - \frac{\partial\tilde\eps_t(\varphi)}{\partial\varphi} \right\| &\leq V\rho^t &\text{and}&&
    \sup_{\theta\in N_\delta} \left\| \frac{\partial h_t^2(\vpt)}{\partial\theta}
    - \frac{\partial\tilde{h}_t^2(\vpt)}{\partial\theta} \right\| &\leq V\rho^t
\end{align*}
in a similar fashion to obtain (\ref{closeness of eps_t and
tilde})--(\ref{closeness of h_t^2 and tilde}). Thus, by using the
inequality
\begin{align*}
    &\left\| \frac{\partial h_t(\vpt)}{\partial\theta} - \frac{\partial\tilde{h}_t(\vpt)}{\partial\theta} \right\| \\
    &\leq {1\over 2h_t(\vpt)} \left\| \frac{\partial h_t^2(\vpt)}{\partial\theta} \right\| \left| h_t(\vpt) - \tilde{h}_t(\vpt) \right|
    + {1\over2} \left\| \frac{\partial h_t^2(\vpt)}{\partial\theta}
    - \frac{\partial\tilde{h}_t^2(\vpt)}{\partial\theta} \right\|,
\end{align*}
we can see that $\sup_{\theta\in N_\delta} \| {\partial
h_t(\vpt)}/{\partial\theta} -
{\partial\tilde{h}_t(\vpt)}/{\partial\theta} \| \leq V\rho^t$, which
ensures {\bf(N4)}(ii). Further, by using similar arguments to verify
(\ref{bddness of tilde eps_t}) and {\bf(N3)}(iii), one can easily
check that {\bf(N4)}(iii) holds. Finally, {\bf(N5)} is a direct
result of Lemma~\ref{positive definiteness}. Therefore, the
asymptotic normality is asserted by Theorem~\ref{asymptotic
normality}. This completes the proof.
\hfill \fbox{} \\


\end{document}